\documentclass[a4paper,12pt]{article}
\usepackage [english]{babel}
\usepackage{amsmath}
\usepackage{amsfonts}
\usepackage{amsthm}
\usepackage{dsfont}
\usepackage{amscd}
\usepackage {geometry}
\usepackage {graphicx}
\usepackage[T1]{fontenc}
\usepackage{nicefrac}
\usepackage{latexsym}
\usepackage[latin1]{inputenc}
\theoremstyle{definiton}
\newtheorem{theorem}{\sc{Th�or�me}}[section]
\newtheorem{theo}[theorem]{Theorem}

\newtheorem{lemm}[theorem]{Lemma}
\newtheorem{lem}{Lemma}
\theoremstyle{remark}

\newcommand{\rA}{\mathcal{A}}

\newcommand{\rE}{\mathcal{E}}
\newcommand{\rL}{\mathcal{L}}

\newcommand{\EE}{\mathbb{E}}

\newcommand{\NN}{\mathbb{N}}

\newcommand{\RR}{\mathbb{R}}

\def\l{\left|}
\def\r{\right|}

\def\x{\chi}

\def\O{\Omega}
\def\F{\Phi}

\def\s{\sigma}

\def\t{\theta}

\def\f{\phi}

\def\o{\omega}
\def\T{\tau}
\def\wh{\widehat}
\def\td{\widetilde}
\def\br{\overline}

\font\timesept=cmr7
\title{Continuous Limits of\\ Classical Repeated Interaction Systems\footnote{Work supported by ANR project "HAM-MARK", N${}^\circ$ ANR-09-BLAN-0098-01}}
\author{Julien DESCHAMPS}
\date{}
\begin{document}

\maketitle

\centerline{\timesept Universite de Lyon}
\vskip -1mm
\centerline{\timesept Universite de Lyon 1, C.N.R.S.}
\vskip -1mm
\centerline{\timesept Institut Camille Jordan}
\vskip -1mm
\centerline{\timesept 21 av Claude Bernard}
\vskip -1mm
\centerline{\timesept 69622 Villeubanne cedex, France}

\begin{abstract}
We consider the physical model of a classical mechanical system (called "small system") undergoing repeated interactions with a chain of identical small pieces (called "environment"). This physical setup constitutes an advantageous way of implementing dissipation for classical systems, it is at the same time Hamiltonian and Markovian. This kind of model has already been studied in the context of quantum mechanical systems, where it was shown to give rise to quantum Langevin equations in the limit of continuous time interactions (\cite{AP}), but it has never been considered for classical mechanical systems yet. The aim of this article is to compute the continuous limit of repeated interactions for classical systems and to prove that they give rise to particular stochastic differential equations in the limit. In particular we recover the usual Langevin equations associated to the action of heat baths. In order to obtain these results, we consider the discrete-time dynamical system induced by Hamilton's equations and the repeated interactions. We embed it into a continuous-time dynamical system and we compute the limit when the time step goes to 0. This way we obtain a discrete-time approximation of stochastic differential equation, considered as a deterministic dynamical system on the Wiener space, which is not exactly of the usual Euler scheme type. We prove the $L^p$ and almost sure convergence of this scheme. We end up with applications to concrete physical exemples such as a charged particule in an uniform electric field or a harmonic interaction. We obtain the usual Langevin equation for the action of a heat bath when considering a damped harmonic oscillator as the small system. 
\end{abstract}

\newpage

\tableofcontents

\section{Introduction}

In order to study open physical systems, that is, systems in interaction with a large environment, two main approaches are often considered. The first one is Hamiltonian, it consists in describing completely the small system, the environment and their interactions in a completely Hamiltonian way. Then one studies the associated dynamical system. The other approach is Markovian, it consists in giving up describing the environment (which is too complicated or unknown) and to describe only the effective action of the environment on the small system. Under some assumptions on the environment, the evolution of the small system is a Markov process. One can then study this Markov process with the associated probabilistic tools (invariant measure, etc).

\smallskip
In the context of quantum mechanical systems, S. Attal and Y. Pautrat propose in \cite{AP} a new type of model for the interaction of a small system and the environment: the \emph{scheme of repeated interactions}. In this setup, the environment is regarded as an infinite assembly of small identical pieces which act independently, one after the other, on the system during a small time step $h$. This approach has the advantage of being in between the two previous approaches: it is Hamiltonian for each interaction of the small system with one piece of the environment is described by a full Hamiltonian, it is also Markovian in its structure of independent and repeated interactions with fresh pieces of the environment. 

This approach, in the quantum context, has also the advantage to give rise to a rather workable way of implementing the dissipation. It is physically realistic for it shown in \cite{AP} that, in the continuous interaction limit ($h$ tends to 0) the associated dynamics converges to the usual quantum Langevin equations associated to open quantum systems. 

Our aim in this article is to consider this scheme of repeated interactions, and its continuous time limit, for classical physical systems.

\smallskip
S. Attal describes in \cite{A} a mathematical framework for classical systems of repeated interactions. His construction is based on a strong connection between Markov processes and dynamical systems that he describes; we present it in Section 2. The main idea is that Markov processes are all obtained from deterministic dynamical systems on product spaces when ignoring one of the two components. In particular stochastic differential equations can be seen as deterministic dynamical systems on the state space times the Wiener space.

The dynamical system associated to repeated classical interactions is discrete in time, depending on the time parameter $h$. If one wants to consider all these dynamical systems for all $h$ and their continuous limit, when $h$ tends to 0, we have to embed discrete time and continuous time dynamical systems into a common natural setup. This is what we develop in Section 3.

\smallskip
Section 4 is devoted to presenting several physical exemples to which our main theorems will be applied at the end of the article.

\smallskip
The convergence of the discrete dynamical systems associated to classical repeated interactions is carefully studied in Section 5. More precisely, the evolution of the system undergoing repeated interactions shall be represented by a Markov chain $(X^h_{nh})$. The study of a limit evolution comes down to the convergence of a linearly interpolated Markov chain to the solution of a stochastic differential equation. Indeed, the embedded dynamics lead to consider the linear interpolation of $(X^h_{nh})$, i.e. the process $(X_t^h)$ defined by
\begin{align*}
 X_t^h = X_{\lfloor t/h\rfloor h}^h +\dfrac{t-\lfloor t/h\rfloor h}{h} \Big{\{}& X_{(\lfloor t/h\rfloor +1)h}^h-X_{\lfloor t/h\rfloor h}^h\Big{\}}\,.
\end{align*}
The main theorems of Section 5 are concerned with the convergence of this process $(X_t^h)$ under some assumptions. Theorem 5.2 shows the $L^{p}$ and almost sure convergences when the evolution of the Markov chain is given by
\[X^h_{(n+1)h}= X_{nh}^h+ \sigma(X_{nh}^h)(W_{(n+1)h}- W_{nh}) + h b(X_{nh}^h) + h \eta^{(h)}(X_{nh}^h,W_{(n+1)h}- W_{nh})\,.\]
The limit process is the solution of the stochastic differential 
\[dX_t = b(X_t)\,dt + \s(X_t)\,d W_t\,\]
when the applications $b$ and $\s$ are Lipschitz, and when $\eta^{(h)}$ satisfies
\[\left|\eta^{(h)}(x,y) \right| \leq K (h^{\alpha}\left|x\right| + \left|y\right|)\,.\]

All these results are then applied in Section 6 to the physical exemples previously  presented in Section 4.

\section{Dynamical Systems and Markov Processes}

\subsection{Discrete Time}

A connexion between deterministic dynamical systems and Markov chains has been recently presented by Attal  in \cite{A}. He shows that randomness in Markov chains appears naturally from deterministic dynamical systems when loosing some information about some part of the system. In this section, we present the context and the main results of \cite{A}.

\smallskip
A \emph{discrete time dynamical system} is a measurable map $\wh{T}$ on a measurable space $(F,\mathcal{F})$. The dynamics of this system is given by the sequence $(\wh{T}^n)_{n \in \mathbb{N}^*}$ of  iterations of $\wh T$. From this map $\wh T$, one can naturally define a map ${T}$ on $\mathcal{L}^{\infty} (F,\mathcal{F})$, the set of bounded functions on $F$, by
\[ {T}f(x)=f(\wh{T}(x))\,,\]
for all $f$ in $\mathcal{L}^{\infty} (F,\mathcal{F})$ and all $x$ in $F$.

\smallskip
Now consider a dynamical system $\wh{T}$ on a product space $S \times E$ where $(S,\mathcal{S})$ and $(E, \rE)$ are two measurable spaces. Furthermore, assume that $(E,\rE)$ is equipped with a probability measure  $\mu$. Physically, $S$ is understood as the phase space of a ``small'' system and $E$ as the one of the environment. Let $T$ be the application on $\mathcal{L}^{\infty} (S \times E)$ induced by $\wh{T}$. 

The idea of the construction developed in \cite{A} is to consider the situation where one has only access to the system $S$ and not to the environment $E$ (for example $E$ might be too complicated, or unknown, or unaccessible to measurement). One wants to understand what kind of dynamics is obtained from $T$ when restricting it to $S$ only.

\smallskip
Consider a bounded function $f$ on $S$, we naturally extend it as a (bounded) function on $S \times E$ by considering
\[(f \otimes \mathds{1}) (x,y)= f(x),\]
for all $x$ in $S$, all $y$ in $E$. That is, the function $f$ made for acting on $S$ only is seen as being part of a larger world $S\times E$. 

The dynamical system $T$ can now act on $f\otimes \mathds{1}$. We assume that what the observer of the system $S$ sees from the whole dynamics on $S\times E$ is an average on $E$ along the given probability measure $\mu$. 
 Therefore,  we have to consider the application $L$ on $\rL^\infty(S)$ defined by
\[Lf(x)= \int_E T(f \otimes \mathds{1}) (x,y)\, d\mu(y)\,.\]
Note that $Lf$ also belongs to $\rL^\infty(S)$. This operator $L$ on $\rL^\infty(S)$ represents the restriction of the dynamics $\wh T$ on $S$.
In other words, we have the following commuting diagram :
\[\begin{CD}
\mathcal{L}^{\infty} (S \times E) @>T>> \mathcal{L}^{\infty} (S \times E)\\
@A{\otimes \mathds{1}}AA @VV{\int_E \cdot \hspace{1mm} d\mu}V\\
\mathcal{L}^{\infty} (S) @>>L> \mathcal{L}^{\infty} (S)
\end{CD}\]
It is natural now to wonder what is the nature of the operator $L$ obtained this way. In \cite{A} the following result is proved.

\begin{theo}
There exists a Markov transition kernel $\Pi$ such that $L$ is of the form 
\[ Lf(x) = \int_S f(z) \,\Pi(x,dz)\,,\]
for all $f\in\rL^\infty(S)$.

\smallskip
Conversely, if $S$ is a Lusin space and $\Pi$ is any Markov transition kernel on $S$, then there exist a probability space $(E,\rE,\mu)$ and a dynamical system $\wh T$ on $S\times E$ such that the operator 
\[ Lf(x) = \int_S f(z) \,\Pi(x,dz)\,,\]
is of the form 
\[Lf(x)= \int_E T(f \otimes \mathds{1}) (x,y)\, d\mu(y)\]
for all $f\in\rL^\infty(S)$.
\end{theo}

The mapping $L$ on the system $S$ is not a dynamical system anymore. It represents the part of the dynamics on the large system $S\times E$ which is obtained when observing $S$ only. The main fact of the above result is that $L$ encodes some randomness, while $T$ was completely deterministic. This randomness arises from the lack of knowledge on the environment.

\smallskip
Note that in the reciprocal above, the dynamical system $T$ which dilates $L$ is not unique.

\bigskip
The Markov operator $L$ obtained above comes from the projection of only one iteration of the dynamical system $T$. It is not true in general that if one projects the mapping $T^k$ on $S$ one would obtain $L^k$ (see \cite{A} for a counter-example).  It would be very interesting to be able to construct a dynamical system $T$ which dilates a Markov operator $L$ not only for one step, but for all their powers $T^k$ and $L^k$. This would mean that we have a whole dynamical system $(T^k)_{k\in\NN}$ which dilates a whole Markov chain $(L^k)_{k\in\NN}$ when restricting it to $S$. This is to say that one wants the following diagram to commute for all $k\in\NN$:
$$
\begin{CD}
\mathcal{L}^{\infty} (S \times E) @>T^k>> \mathcal{L}^{\infty} (S \times E)\\
@A{\otimes \mathds{1}}AA @VV{\int_E \cdot \hspace{1mm} d\mu}V\\
\mathcal{L}^{\infty} (S) @>>L^k> \mathcal{L}^{\infty} (S)
\end{CD}
$$
As explained in \cite{A} this can be obtained through the scheme of repeated interactions, as follows.

Let $\wh{T}$ be a dynamical system on $S \times E$ which dilates a Markov operator $L$. Since $\wh{T}$ is a measurable map from $S \times E$ to $S \times E$, there exist two measurable applications $U$ and $V$ such that, for all $x$ in $S$, all $y$ in $E$,
\[\wh{T}(x,y) = (U(x,y),V(x,y))\,.\]
Now, consider the space $\td{E} = E^{\NN^*}$ endowed with the usual $\s$-field $\rE^{\otimes \NN^*}$ and the product measure $\td{\mu}=\mu^{\otimes \NN^*}$. From the map $\wh{T}$, one defines a dynamical system $\td{T}$ on the space $S \times \td{E}$ by
\[\td{T}(x,y)=(U(x,y_1),\t(y))\,,\]
for all $x$ in $S$, all sequence $y=(y_m)_{m \in \NN^*}$ in $\td{E}$, where $\t$ is the shift on $\td{E}$, that is, 
$$
\t(y)=(y_{m+1})_{m \in \NN^*}\,.
$$
Physically, this construction can be understood as follows. 
The system $S$ is in interaction with 
a large environment. This large environment is a chain $\td{E}$ made of copies of a small piece $E$. One after the other, each part $E$ of this environment interacts with the system $S$, independently from the others. This is the so-called ``\emph{repeated interactions scheme}''.

\smallskip
Note that the application $V$ doesn't appear in the definition of $\td{T}$.
From the physical point of view this is natural, the map $V$ gives the evolution of the part $E$ of the environment which acts on the system. As this part shall not be involved in the dynamics of the system $S$, its new state has no importance for the system.

\smallskip
As previously, the mapping $\td{T}$ induces an operator $T$ on $\rL^{\infty}(S \times \td{E})$. 
Then, the following theorem (proved in \cite{A}) shows that the sequence $(T^k)$ of iterations of $T$ dilates the whole Markov chain $(L^k)_{k\in\NN}$.

\begin{theo}
For all $m$ in $\mathbb{N}^*$, all $x$ in $S$, and all $f$ in $\mathcal{L}^{\infty} (S)$, 
\[(L^m f) (x)= \int_{\td{E}} T^m(f \otimes \mathds{1}) (x,y)\, d\td{\mu}(y)\,.\]
\end{theo}

The operator $L$ represents the action of a Markov kernel $\Pi$ on bounded applications and the restriction of the whole dynamics given by $T$ to $S$.
On the other hand, the Markovian behaviour on applications associated to the action of $\Pi$ can be also seen on points when focusing on the dynamical system $\td{T}$. Indeed, for all initial state $x$ of the system in $S$, and all state of the environment $y$ in $\td{E}$, the evolution of the system is given by the sequence $(X_n^x(y))_{n \in \NN}$ defined by \[X_{n+1}^x (y) = U(X_n^x(y), y_{n+1})\,,\]
with $X_0=x$.
Now if the state of the environment is unknown, the dynamics of the system is represented by the sequence $(X_n^x)_{n \in \NN}$ of applications from $\td{E}$ to $S$.
This sequence $(X_n^x)$ is a Markov chain. Furthermore, the Markov transition kernel of $(X_n^x)$ 
 is $\Pi$ again.

\smallskip
Note that the whole dynamics of $\td{T}$ can be expressed from this sequence $(X_n^x)$ as follows. For all $k$ in $\NN$, all $x$ in $S$, and all $y$ in $\td{E}$,
\[\td{T}^k(x,y)=(X^x_k(y),\t^k(y))\,,\]
where $\t^k(y)$ is the sequence $y$ which is shifted $k$ times, i.e. $\t^k(y)=(y_{n+k})_{n \in \NN^*}$.

\bigskip
Note that sequences are indexed by $\NN^*$ in this scheme of repeated interactions. Physically, this can be understood as repeated interactions between the systems which last for a time duration $1$. Now, a new parameter $h$ is added. It represents the time step for the interactions. Henceforth, all the sequences are now indexed by $h \NN^*$. 

For instance, the dynamics of the system is now given by the Markov chain $(X_{nh}^{(h)})_{n \in \NN}$ defined for all starting state $x$ by 

\[X_{(n+1)h}^{(h)} (y) = U^{(h)}(X_{nh}^{(h)}(y), y_{(n+1)h})\,,\] 
with $X_0^{(h)}=x$.

The dynamical system given by the scheme of repeated interactions is now 
\[\td{T}^{(h)}(x,y)=(U^{(h)}(x,y_h),\t^{(h)}(y))\,,\]
where $\t^{(h)}$ is the shift, i.e. $\t^{(h)}(y)=(y_{(n+1)h})_{n \in \NN^*}$.

\bigskip
Note that the application $U^{(h)}$ can also depend on the time step $h$. We shall see in the exemples of Section 4 explicit expressions for this map $U^{(h)}$ in some physical situations.

\smallskip
Our aim is now to understand what can be the limit dynamics of these repeated interactions when the time step $h$ goes to $0$.

Attal and Pautrat, in \cite{AP}, have studied open quantum systems with an equivalent setup of repeated interactions. They show the convergence of repeated interactions to quantum stochastic differential equations.
In section 5, we shall see that the limit evolutions in the classical case are solutions of some stochastic differential equations. Therefore, we shall need to extend our parallel between Markov chains and dynamical systems to the continuous time setup and in particular to solutions of stochastic differential equations. 

\subsection{Continuous Time}

As previously seen, from a Markov operator $L$, with the help of repeated interactions, one can construct a dynamical system which dilates the whole sequence~$(L^k)$. Moreover, the evolution of the first component is given by a Markov chain associated to the Markov transition kernel. We wish to extend this idea to Markov processes which are solutions of stochastic differential equations.

\smallskip
A \emph{continuous time dynamical system} is a semigroup of measurable applications $(T_t)_{t \in \RR^+}$, that is, which satisfy $T_s \circ T_t = T_{s+t}\,,$ for all $s$, $t$ in $\RR^+$.

Consider now a stochastic differential equation (SDE) on $\RR^m$,
\begin{equation}
dX_t = b(X_t)\, dt + \s(X_t) \,dW_t \label{eq:A}\,,
\end{equation}
where $W_t$ is a $d$-dimensional standard Brownian motion, $b$ and $\s$ are measurable functions respectively from $\mathbb{R}^m$ to $\mathbb{R}^m$, and from $\mathbb{R}^m$ to $\mathcal{M}_{m,d} (\mathbb{R})$, the space of $m \times d$ real matrices.
We want to create an environment $\O$ and a deterministic dynamical system $(T_t)$ on $\mathbb{R}^m \times \O$ which dilates the solution and such that the first component of $T_t$ is the solution of
\eqref{eq:A} at time $t$.

\smallskip
However, we shall need existence and uniqueness of the solution of \eqref{eq:A} at all time $t$ and for all initial conditions. In order to guarantee that properties we have to make some assumptions on the functions $b$ and $\s$. We require them to be either globally Lipschitz,
\begin{enumerate}
\item[(H1)]There exists $K >0$ such that for all $x,y$ we have
\[\left|b(x)-b(y)\right|\leq K_0 \left|x-y\right|\, \text{ and}  \hspace{5mm} \|\sigma(x)-\sigma(y)\| \leq K_0 \left|x-y\right|\,,\]
\end{enumerate}
or locally Lipschitz and linearly bounded,
\begin{enumerate}
\item[(H2)] The functions $b$ and $\sigma$ are locally Lipschitz:
for all $N>0$, there exists $K_N >0$ such that for all $x,y \in \mathbb{B}(0,N)$ we have
\[\left|b(x)-b(y)\right|\leq K_N \left|x-y\right|\, \text{ and}  \hspace{5mm} \|\sigma(x)-\sigma(y)\| \leq K_N \left|x-y\right|\,.\]
\item[(H3)] Linear growth bound: There exists a constant $K_1>0$ such that 
\[\left|b(x)\right| \leq K_1(1 + \left|x\right|)\qquad\mbox{and}\qquad
\|\sigma(x)\|\leq K_1(1 + \left|x\right|)\,,\]
\end{enumerate}
where $\left| \cdot \right|$ is the euclidean norm on $\RR^m$ and $\| \cdot \|$ is the Hilbert-Schmidt norm on $\mathcal{M}_{m,d} (\mathbb{R})$ defined by
\[\| \s \|^2 = \sum_{j = 1}^d \sum_{i=1}^m (\s^{i,j})^2\,.\]
Note that Assumption (H1) implies (H2) and (H3). We differentiate the two cases because the convergences studied in section 5 shall be stronger for globally Lipschitz functions $b$ and $\s$.
 
Note also that Assumption (H1) or Assumptions (H2) and (H3) are sufficient (see \cite{IW}) but not necessary.
\smallskip
One can now construct the environment and the dynamical system. Consider the Wiener space $(\O,\mathcal{F},\mathbb{P})$ associated to the canonical Brownian motion $(W_t)$. This is to say that $\O=C_0 (\mathbb{R}_+, \mathbb{R}^d)$ is the space of continuous functions from $\mathbb{R}_+$ to $\mathbb{R}^d$ vanishing at the origin. The canonical Brownian motion $(W_t)$ is then defined by $W_t (\omega)=\omega(t)$, for all $\omega\in\O$ and all $t\in\mathbb{R}_+$.

Consider the \emph{shift} $\theta_t$ defined on $\O$ by
\[\theta_t(\omega)(s)=\omega (t+s) - \omega(t)\,.\]
and the family of mappings $(T_t)_{t\in \RR^+}$ on $\RR^+\times \O$ defined by
\[T_t(x,\omega) = (X^x_t(\omega),\theta_t(\omega))\,,\]
for all $x\in\RR^m$, all $\o\in\O$,  where $X^x_t$ is the solution of \eqref{eq:A} at time $t$, starting at $x$. In \cite{A}, Attal shows the following.

\begin{theo}
The family $(T_t)$ on $\mathbb{R}^m \times \O$ is a continuous time dynamical system.
\end{theo}
Let us denote by $\mathcal{T}_t$ the map induced by $T_t$ on $\mathcal{L}^\infty(\RR^m)$. It induces a semigroup $(P_t)_{t \in \RR^+}$ on $\mathcal{L}^\infty(\RR^m)$ by 
\[P_t(f)(x)=\mathbb{E}\left[\mathcal{T}_t(f \otimes \mathds{1})(x, \cdot )\right]\,.\]
The generator of this semigroup $(P_t)$ is
 \[\rA= \frac{1}{2} \sum_{i,j=1}^d a^{i,j}(x) \frac{\partial^2}{\partial x^i \partial x^j} + \sum_{i=1}^d  b^i(x) \frac{\partial}{\partial x^i}\,,\]
where $a$ is the symmetric matrix $\sigma \sigma^t$. 

\section{Embedding Discrete Time into Continuous Time}

In order to prove the convergence of the discrete time dynamical systems, associated to repeated interactions, to the continuous time dynamical systems associated to solutions of stochatic differential equations, we need to explicitly embed the discrete time dynamical systems into a continuous time setup.

The state space on which the repeated interaction dynamical system $\td{T}^{(h)}$ acts is $\RR^m \times (\RR^d)^{h\NN^*}$, whereas the one of $(T_t)_{t \in \RR_+}$ is $\RR^m \times \O$. The first step in constructing the embedding of dynamical systems is to construct an embedding of $(\RR^d)^{h\NN^*}$ into $\O$.

\subsection{Discrete Approximation of $\O$}

Let $\f_I^{(h)}$ be the map from $(\mathbb{R}^d)^{h \mathbb{N}^*}$ to $\O$  defined by
\[\f_I^{(h)}(y) (t) = \sum_{n=0} ^{\lfloor t/h \rfloor} y_{nh} + \frac{t - \lfloor t/h \rfloor h}{h} y_{(\lfloor t/h \rfloor +1)h}\, ,\]
where $y_0 =0$.
This map $\f_I^{(h)}$ actually builds a continuous, piecewise linear, function whose increments are the elements of the sequence $y$. The range of $\f_I^{(h)}$ is denoted by $\O^{(h)}$, it is a subspace of $\O$.

\smallskip
Conversely, define the map $\f_P^{(h)}$ from $\O$ to $(\mathbb{R}^d)^{h \mathbb{N}^*}$ by
\[\f_P^{(h)}(\omega)=(W_{nh}(\omega) - W_{(n-1)h}(\omega))_{n \in \mathbb{N}^*}= (\omega (nh) - \omega((n-1)h))_{n \in \mathbb{N}^*} \,.\]
In other words, the range of an element of $\O$ by $\f_P^{(h)}$ is the sequence of its increments at the times $nh$, for all $n\in\NN^*$.
Note that these applications $\f_I^{(h)}$ and $\f_P^{(h)}$ satisfy 
\[ \f_P^{(h)} \circ \f_I^{(h)} = \mbox{Id}_{(\mathbb{R}^d)^{h \mathbb{N}^*}}\,.\]
In particular the map $\f_I^{(h)}$ is an injection from $(\mathbb{R}^d)^{h \mathbb{N}^*}$ to $\O$, the spaces $(\mathbb{R}^d)^{h \mathbb{N}^*}$ and $\O^{(h)}$ are in bijection through the applications $\f_I^{(h)}$ and $(\f_P^{(h)})_{\vert \O^{(h)}}$.

\smallskip
We now show that $(\mathbb{R}^d)^{h \mathbb{N}^*}$ can be viewed as an approximation of $\O$ for the usual metric associated to the topology of uniform convergence on compact sets. For two elements $\o$ and $\o'$ in $\O$ define the distance
\[D(\omega,\omega') = \sum_{n=1}^{\infty} \frac{1}{2^n} \hspace{3mm} \dfrac{\underset{0\leq t \leq n}{\sup} \left| \omega(t) - \omega '(t) \right|}{1+\underset{0\leq t \leq n}{\sup} \left| \omega(t) - \omega '(t) \right|}\,,\]
where $\left| \cdot \right|$ is the euclidean norm on $\mathbb{R}^d$.
The space $\Omega$ endowed with this metric is a polish space.

\begin{lem}
For all $\omega \in \Omega$,  \[\underset{h \rightarrow 0}{\lim} \hspace{2mm} D(\omega,\f_I^{(h)} \circ \f_P^{(h)} (\o) )=0\,.\]
\end{lem}

\begin{proof}
This result is based on the uniform convergence of piecewise linear applications to continuous one on compact sets.
Let $\omega$ be a function in $\Omega$. Then, by definition,
$$
\f_P^{(h)}(\omega)=(\omega (nh) - \omega((n-1)h))_{n \in \mathbb{N}^*}\,.
$$
Now the mapping $\f_I^{(h)}$ is applied to sequence $\f_P^{(h)}(\omega)$. For all $t$ in $\RR_+$, we have
\begin{multline*}
\f_I^{(h)} \circ \f_P^{(h)} (\o) (t)= \sum_{n=0} ^{\lfloor t/h \rfloor} (\omega (nh) - \omega((n-1)h)) +\hfill \\
\hfill +\frac{t - \lfloor t/h \rfloor h}{h} \Big{\{}\o((\lfloor t/h \rfloor +1)h)- \o(\lfloor t/h \rfloor h)\Big{\}}\\
\hphantom{\f_I^{(h)} \circ \f_P^{(h)} (\o) (t) \ \ \ }=\,\o(\lfloor t/h \rfloor h)  + \frac{t - \lfloor t/h \rfloor h}{h} \Big{\{}\o((\lfloor t/h \rfloor +1)h)- \o(\lfloor t/h \rfloor h)\Big{\}} \,.\hfill
\end{multline*}
Note that $\f_I^{(h)} \circ \f_P^{(h)} (\o) (t)$ is a point in the segment
$
\Big{[}\o(\lfloor t/h \rfloor h))\,,\o((\lfloor t/h \rfloor +1)h)\,\Big{]}
$
Since $\o$ is continuous we have
\[\lim_{h \rightarrow 0}\o(\lfloor t/h \rfloor h)=\lim_{h \rightarrow 0}\o((\lfloor t/h \rfloor +1)h)= \o(t)\,.\]
Therefore,
\[ \lim_{h \rightarrow 0} \f_I^{(h)} \circ \f_P^{(h)} (\o) (t)= \o (t)\,.\]
Let $n\in\NN$, as $\left[0,n\right]$ is compact the function $\o$ is uniformly continuous on this interval.
Thus,
\[ \lim_{h \rightarrow 0}\ \sup_{0\leq t \leq n} \left| \omega(t) - \f_I^{(h)} \circ \f_P^{(h)} (\o) (t) \right| = 0\,.\]
On the other hand, note that 
\[\frac{\underset{0\leq t \leq n}{\sup} \left| \omega(t) - \f_I^{(h)} \circ \f_P^{(h)} (\o) (t) \right|}{1+\underset{0\leq t \leq n}{\sup}  \left| \omega(t) - \f_I^{(h)} \circ \f_P^{(h)} (\o) (t) \right|} \leq 1\,.\] 
Hence, by Lebesgue's theorem,
\[ \lim_{h \rightarrow 0} D(\omega,\f_I^{(h)} \circ \f_P^{(h)} (\o)) = 0\,.\]
\end{proof}

\subsection{Embedding the Discrete Time Dynamics}

The first step in the construction of a continuous dynamics from $\td{T}^{(h)}$ is to relate it to a discrete time evolution on $\RR^m \times \O$.
Since the applications $\f_I^{(h)}$ and $\f_P^{(h)}$ are only defined on $(\mathbb{R}^d)^{h \mathbb{N}^*}$ or $\O$, we extend them to $\RR^m \times (\mathbb{R}^d)^{h \mathbb{N}^*}$ and  $\RR^m \times \O$, respectively, by
\[\F_I^{(h)} = (Id,\f_I^{(h)})\,,\qquad \F_P^{(h)} = (Id,\f_P^{(h)}) \,.\]
Consider the dynamical system $\overline{T}^{(h)}$ on $\RR^m \times \O$ given by
\[\overline{T}^{(h)}= \F_I^{(h)} \circ \td{T}^{(h)} \circ \F_P^{(h)}\,.\]
As we have
\[ \F_P^{(h)} \circ \F_I^{(h)} = \mbox{Id}_{\RR^{m} \times \O}\,,\]
the following diagram is communting for all $n \geq 1$:
\[\begin{CD}
\mathbb{R}^m \times \Omega @>(\overline{T}^{(h)})^n>> \mathbb{R}^m \times \Omega \\
@V\F_P^{(h)}VV @AA\F_I^{(h)}A\\
\mathbb{R}^m \times (\mathbb{R}^d)^{h \mathbb{N}^*} @>>(\td{T}^{(h)})^n> \mathbb{R}^m \times (\mathbb{R}^d)^{h \mathbb{N}^*}
\end{CD}\]
We now relate this dynamical system to a continuous time dynamics, by linearly interpolating in time. Define a new family of applications $(\overline{T}^{(h)}_t)_{t \in \mathbb{R}^+}$ by 
\begin{equation*}
\begin{split}
\overline{T}^{(h)}_t &=(\overline{T}^{(h)})^{\lfloor t/h\rfloor}+ \dfrac{t-\lfloor t/h\rfloor h}{h}\hspace{3mm} \Big{\{}(\overline{T}^{(h)})^{(\lfloor t/h\rfloor+1) } - (\overline{T}^{(h)})^{\lfloor t/h\rfloor}\Big{\}}\\
&=\frac{(\lfloor t/h\rfloor+1)h - t}{h}\hspace{3mm} (\overline{T}^{(h)})^{\lfloor t/h\rfloor } + \frac{t-\lfloor t/h\rfloor h}{h}\hspace{3mm} (\overline{T}^{(h)})^{(\lfloor t/h\rfloor+1)}\,,
\end{split}
\end{equation*}
where, by convention, $(\overline{T}^{(h)})^0=\mbox{Id}$. 

The projections on the first and the second component of $\overline{T}^{(h)}_t$ are respectively denoted by $\overline{X}^{(h)}_t$ and $\overline{\t}_t^{(h)}$.
More precisely, for all initial $x \in \RR$ and $\o$, the value of $\br{T}^{(h)}_t(x,\o)$ can be also expressed as

\[\br{T}^{(h)}_t(x,\o)= (\br{X}^{(h)}_t (\f_P^{(h)}(\o)),\br{\theta}^{(h)}_t(\o))\,.\]

\smallskip
Note that, in general, the family $(\br{T}^{(h)}_t)_{t \in \mathbb{R}_+}$ is not a semigroup because of the linear interpolation. The random process $\br{X}^{(h)}_t$ is not also Markovian but a linearly interpolated Markov chain.

Also note that the state of the environment for the evolution of the system is given by $\f_P^{(h)}(\o)$, i.e. by the increments of a continuous application of the Wiener space $\O$.

Finally, the convergence of the dynamical system $\td{T}^{(h)}$ to a continuous one can be studied by examining the dynamics of $(\br{T}^{(h)}_t)_{t \in \mathbb{R}_+}$, more exactly the convergence of the random process $\br{X}^{(h)}$ to a solution of a stochastic differential equation and the convergence of $\bar{\t}_t^{(h)}$ to the shift $\t_t$ on $\O$ according to the metric $D$.

\smallskip
However, before hands, we want to illustrate this framework and in particular the scheme of repeated interactions, through physical examples.

\section{Application to some Physical Systems}

As explained in the introduction, the main motivation for the study of repeated interaction schemes and their continuous limit is to try to obtain physically justified and workable models for the dissipation of a simple system into a large environment, such as a heat bath for example.

\subsection{Charged Particle in a Uniform Electric Field}

Consider a particle of charge $q$ and mass $m$ in an uniform electric field $E$ in dimension $1$. Its energy without interaction with $E$ is just kinetic, i.e. $p^2/2m$. In the presence of the exterior electric field, the particule has a potential energy $-qxE$, where $x$ is the position.
Thus, the Hamiltonian of the particle is
\[H(x,p) = \frac{p^2}{2m} -qxE\,.\]
The dynamics of the particle is governed by Hamilton's equations of motion,
\[\begin{cases}
\dot{x}= \dfrac{p}{m}\\
\dot{p}= qE\,.
  \end{cases}\]
These equations can be easily solved and give
\[\begin{cases}
x(t)= \dfrac{qE}{2m} t^2+\dfrac{p(0)}{m}t+ x(0)\\
p(t)= qE t + p(0)\,.
  \end{cases}\]
Now we set up the scheme of repeated interactions. The small system is the particle. As it moves in dimension $1$, the space $S$ is $\RR^2$ endowed with its Borelian $\s$-algebra.
The environment is the exterior electric field. Hence, the space $E$ is $\RR$. Intuitively, the limit process shall be a solution of a SDE driven by a $1$-dimensional Brownian motion $(W_t)$.

The environment of the discrete time dynamics at time step $h$ is the set $\RR^{h \NN^*}$.
The interaction between the system and the environment is described as follows. At each time $(n-1)h$, the value of the electric field $E(nh)$ is sampled from the increment $W(nh)-W((n-1)h)$ of the $1$-dimensional Brownian motion. The system evolves during a time $h$ according to the solution of the equations whose initial values $x((n-1)h)$ and $p((n-1)h)$ and where the electric field is $E(nh)$. After this time $h$, the interactions are stopped and one repeats the procedure.

In particular the evolution of the system is given by the following Markov chain
\[\begin{cases}
x(nh)= x((n-1)h)+h\dfrac{p((n-1)h)}{m}+ h^2\dfrac{qE(nh)}{2m}\\
p(nh)=p((n-1)h)  + h qE(nh)\,.
  \end{cases}\]
  
However, one has to make some renormalization somewhere. Indeed, when the time step $h$ decreases, the effect of the electric field on the particle becomes smaller and smaller. To counter that, the interactions have to be reinforced. The electric field $E$ is renormalized by multiplying by a factor $1/h$. 

This normalization factor can be intuitively understood as follows. First if one wants to keep the same intensity (at least in law) for the interactions, a factor $1/\sqrt{h}$ is needed. On the other hand, one needs to renormalize with an other $1/\sqrt{h}$ as in the quantum case of \cite{AP}, since the influence of the environment decreases with the time step $h$. Thus, the value of the electric field is now sampled from $1/{h}\,\big(W(nh)-W((n-1)h)\big)$.

Therefore, the new dynamics of the system is 
\[\begin{cases}
x(nh)= x((n-1)h)+h\dfrac{p((n-1)h)}{m}+ h\dfrac{qE(nh)}{2m}\\
p(nh)=p((n-1)h)  + qE(nh)\,.
  \end{cases}\]
In other words, using vector notations, the dynamics is
\[X(nh)= U^{(h)}(X((n-1)h),E(nh))\,,\]
where the map $U^{(h)}$ is defined by

\[ U^{(h)}(x,y)= x +  \s(x) y +h b(x) + h\eta^{(h)}(x,y)\,,\]
with
\[\s \left(\begin{array}{c} x_1\\ x_2 \end{array}\right)=\left(\begin{array}{c} 0\\ q \end{array}\right)\,, \hspace{10mm}b\left(\begin{array}{c} x_1\\ x_2 \end{array}\right)=\left(\begin{array}{c} \dfrac{x_2}{m}\\ 0 \end{array}\right)\,,\hspace{10mm}\eta^{(h)}(x,y)=\left(\begin{array}{c} \dfrac{qy}{2m}\\ 0 \end{array}\right)\,.\]

\subsection{Harmonic Interaction}

Our second example is another example where Hamilton's equations can be explicitly solved.

Consider two unit mass systems linked by a spring whose spring constant is $1$. Let $l$ be the length for which the potential energy is minimum.
We assume that the two objects can just horizontaly move without friction.

In this case, the Hamiltonian of the two objects is 
\[H \left[ \left(\begin{array}{c}
Q_1\\
P_1
\end{array}\right), \left(\begin{array}{c}
Q_2\\
P_2
\end{array}\right) \right] =\frac{P_1^2}{2} + \frac{P_2^2}{2} + \dfrac{1}{2}( Q_2 - Q_1 - l)^2 \,,\]
where ${P_1^2}/{2}$ is the kinetic energy of the system $1$ and ${P_2^2}/{2}$ the kinetic energy of the system $2$.

The dynamics of the whole system is given by Hamilton's equations
\[\begin{cases}
\dot{Q_1} = \dfrac{\partial H}{\partial P_1}  = P_1\qquad\mbox{and}\qquad
\dot{P_1} = -\dfrac{\partial H}{\partial Q_1} =  Q_2 -  Q_1 -l\\
\\
\dot{Q_2} = \dfrac{\partial H}{\partial P_2}  = P_2\qquad\mbox{and}\qquad
\dot{P_2} = -\dfrac{\partial H}{\partial Q_2} =   -  (Q_2 -Q_1-l)\,.
  \end{cases}
\]
These equations can be solved and give this evolution of the whole system according to the initial conditions.
\begin{align*}
 Q_1(t) =& \dfrac{1}2\left(P_1(0)+P_2(0)\right) t +\dfrac12(Q_1(0) + Q_2(0)- l)+ \\
 &\ \ \ +\dfrac12(Q_1(0)-Q_2(0)+l)\, \cos(\sqrt{2}t) +\dfrac1{2\sqrt2}(P_1(0)-P_2(0))\,\sin(\sqrt{2}t) \\
 P_1(t) =& \dfrac12(P_1(0)+P_2(0))- \dfrac1{\sqrt2}(Q_1(0)-Q_2(0)+l)\, \sin(\sqrt{2}t) + \\
 &\ \ \ +\dfrac12(P_1(0)-P_2(0))\, \cos(\sqrt{2}t)
 \end{align*}
 \begin{align*}
Q_2(t)  =& \dfrac12(P_1(0)+P_2(0)) t +\dfrac12(Q_1(0)+Q_2(0)+ l)+ \\
&\ \ \ -\dfrac12(Q_1(0)-Q_2(0)+l)\, \cos(\sqrt{2}t) - \dfrac1{2\sqrt2}(P_1(0)-P_2(0))\,\sin(\sqrt{2}t)\\
 P_2(t) =& \dfrac12(P_1(0)+P_2(0)) + \dfrac1{\sqrt2}(Q_1(0)-Q_2(0)+l)\,\sin(\sqrt{2}t)+ \\
 &\ \ \ -\dfrac12(P_1(0)-P_2(0))\, \cos(\sqrt{2}t)\,. 
\end{align*}

Let $h$ be a fixed time step. As $h$ is supposed sufficiently small,
an approached expression of the discrete time dynamical system can be computed by a Taylor expansion. 

On the other hand, as we shall be only interested in the dynamics of the system 1, the evolutions of $Q_2$ and $P_2$ are forgotten.
We get
\begin{multline*}
 Q_1(h) = Q_1(0)+h P_1(0) + \dfrac{1}{2}\left(Q_2(0)-Q_1(0)-l\right)h^2+ \hfill \\
 \hfill - \dfrac{1}{6}\left(P_1(0)-P_2(0)\right) h^3  +\circ(h^3)  \\
\ \ P_1(h) =  P_1(0) +\left(Q_2(0)-Q_1(0)-l\right) h+ \dfrac{1}{2} \left(P_2(0)-P_1(0)\right) h^2+\hfill \\
\hfill+ \dfrac{1}{3}\left(Q_1(0)-Q_2(0)+l\right) h^3+\circ(h^3)\,.
\end{multline*}
We now introduce the scheme of repeated interactions.
System $1$ is chosen as the small system. Thus, the space $S$ is $\RR^2$.
System $2$ plays the role of one piece of the environment. The space $E$ is also in this case $\RR^2$.

At time $nh$, the small system is in the state $Q_1(nh)$ and $P_1(nh)$. The values of $Q_2((n+1)h)$ and $P_2((n+1)h)$ are sampled from the increments of a $2$-dimensional Brownian motion $W_t$. One makes the system and the environment interact during a time $h$ with the initial conditions for the environment $Q_2((n+1)h)$ and $P_2((n+1)h)$. The interaction is stopped after a time $h$. One then repeats the procedure.

\smallskip
However, for the same reasons as for the first example, interactions need to be renormalized by a factor ${1}/{h}$.
The Markov chain which gives the evolution of the system becomes
\begin{multline*}
 Q_1(nh) = Q_1((n-1)h)+ \Big(P_1((n-1)h)  + \dfrac{1}{2} Q_2(nh)\Big) h +\hfill \\
  \hfill- \dfrac{1}{2} \Big(Q_1((n-1)h)+l - \dfrac{P_2(nh)}{3}\Big) h^2+ \circ(h^2) \\
 \ \ P_1(nh) = P_1((n-1)h) + Q_2(nh)-\Big(Q_1((n-1)h)+l -\dfrac{1}{2} P_2(nh)\Big) h +\hfill \\ 
\hfill - \Big(\dfrac{P_1((n-1)h)}{2} + \dfrac{Q_2(nh)}{3}\Big)h^2 + \circ(h^2)\,,
\end{multline*}
or, equivalently
\[X(nh)=U^{(h)}(X((n-1)h), Y(nh))\,\]
where
\begin{align*}
U^{(h)}(X,Y)= X + \sigma(X) Y + h b(X) +h \eta^{(h)}(X,Y),
\end{align*}
with
\[b\left(\begin{array}{c} x_1\\x_2  \end{array} \right) = \left(\begin{array}{c} x_2\\ -(x_1+l) \end{array} \right), \hspace{5mm} \sigma \left(\begin{array}{c} x_1\\ x_2 \end{array} \right) =\left(\begin{array}{cc} 0& 0\\ 1& 0 \end{array} \right)\,,\]
and
\[\eta^{(h)}\left[\left(\begin{array}{c} x_1\\ x_2 \end{array}\right), \left(\begin{array}{c} y_1\\ y_2 \end{array}\right) \right] = \dfrac{1}{2}\left(\begin{array}{c} y_1\\ y_2 \end{array}\right)-\dfrac{h}{2} \left(\begin{array}{c} x_1+l-y_2/3  \\ x_2+2y_1/3 \end{array}\right) + \circ(h)\,.\]

Note that this application $U^{(h)}$ is of the same form as in the case of the particle in an uniform electric field.

\bigskip
Note also that the renormalization can be  understood in a more general way as follows. In the two examples above, the interactions are reinforced by changing the states of the environment in the scheme of repeated interactions. But, in the same way as in the quantum case, the renormalization can be understood as a reinforcement of the interaction in the Hamiltonian.

Indeed, the Hamiltonien of the whole system can be written in general way, 
\[H  =\underbrace{H_1}_{System}+ \underbrace{H_2}_{Environment} + \underbrace{I}_{Interaction},\]
where $H_1$ is the part of the Hamiltonian which only depends on the state of the system ( ${P_1^2}/{2} +{Q_1^2}/{2}$ for instance in the case of the harmonic interaction), $H_2$, the part which only depends on the state of the environment (${P_2^2}/{2} +{Q_2^2}/{2}$) and the interaction part $I$ which really depends on the two states ($- Q_1 Q_2$). Our renormalization factor can be viewed as a reinforcement of the interaction term $I$ by multiplying it by this factor ${1}/{h}$. Then, the scheme of repeated interactions is set up from this new Hamiltonian which depends on $h$.

\subsection{Damped Harmonic Oscillator}

The formalism developped in Section 2 is general and it can also be applied to non-Hamiltonian systems. An example based on a change of the harmonic interaction by adding a friction term for the system is presented in this part.

Consider the same system as previously, i.e. two harmonically interacting objects of mass $1$. We assume that the system 1 also undergoes a fluid friction in $-f P_1$ where $f$ is the friction coefficient. Because of this force, the energy is not conserved, and therefore, the system is not Hamiltonian. Moreover, in order to simplify the interaction, the length $l$ of the minimum of the potential energy is supposed to be $0$.

The evolution of the system follows Newton's law of motion,  
\[\begin{cases}
\dot{Q_1}= P_1\\
\dot{P_1}= - f P_1 + Q_2-Q_1\,.
  \end{cases}\]
On the other hand, the environment part stays Hamiltonian and evolves according to the equations 
\[\begin{cases}
\dot{Q_2}= P_2\\
\dot{P_2}= Q_1-Q_2\,.
  \end{cases}\]
For a small time $h>0$, Taylor's expansions and the previous equations lead to a state of the system after a time $h$  
\[\begin{cases}
Q_1(h)=Q_1(0) + h P_1(0)+ {\cal O} (h^2)\\
P_1(h)= P_1(0) +h(- f P_1(0) + Q_2(0)-Q_1(0))+{\cal O}(h^2)\,.
  \end{cases}\]

We now set up the repeated interactions framework.
The space of the system is $\RR^2$. The environment is represented by the chain $(\RR^2)^{h\NN^*}$.
The motion of the system is given by the following Markov chain
\[\begin{cases}
Q_1((n+1)h)=Q_1(nh) + h P_1(nh)+ {\cal O} (h^2)\\
P_1((n+1)h)= P_1(nh) +h(- f P_1(nh) + Q_2(nh)-Q_1(nh))+{\cal O}(h^2)\,.
  \end{cases}\]
The sequence $(Q_2(nh), P_2(nh))_{n \in \NN}$ is sampled from the increments of a $2$-dimensional Brownian motion. However, as previousy, states of the environment are reinforced by a factor $1/h$.
Hence, the Markov chain $(X(nh))$ is defined by
\[X(nh)=U^{(h)}(X((n-1)h), Y(nh))\,\]
where
\begin{align*}
U^{(h)}(X,Y)= X + \sigma(X) Y + h b(X) +h \eta^{(h)}(X,Y),
\end{align*}
with
\[b\left(\begin{array}{c} x_1\\x_2  \end{array} \right) = \left(\begin{array}{c} x_2\\ -x_1 - f x_2 \end{array} \right), \hspace{5mm} \sigma \left(\begin{array}{c} x_1\\ x_2 \end{array} \right) =\left(\begin{array}{cc} 0& 0\\ 1& 0 \end{array} \right)\,,\]
and the remaining terms of Taylor's expansion are grouped together in the application $\eta^{(h)}$.

Note that this last function could be explicitely determined from the equations of Newton's law of motion from which higher derivatives can be expressed.

\section{Convergence of Dynamics}
We leave for the moment our physical examples and come back to the general setup as introduced in Sections 2 and 3.
In Subsection 3.2 a continuous dynamics $\br{T}^{(h)}$ related to $\td{T}^{(h)}$ was defined on $\RR^m \times \O$.
The convergence of $\br{T}^{(h)}$ to a continuous time dynamical system like $(T_t)$ is now studied in this part. 

Each dynamics acts on a product space. Therefore, we examine separately the functions on each component. More precisely, the convergence of the application $\br{\t}_t^{(h)}$ to $\t_t$ for all $t$ on $\O$ is firstly regarded in the next subsection. Then, the one of the processes $\br{X}_t$ to a solution of stochastic differential equation depending particularly on the application $U^{(h)}$ is studied. 

\subsection{Convergence of Shift}

The space of the environment is $\O$, the set of continuous functions from $\RR_+$ to $\RR^d$ vanishing at the origin. If the limit dynamical system exists, the space on which the shift acts has to be $\O$ too. Therefore, we now consider the shift $\t_t$ on $\O$. Recall that it is defined by 
\[\theta_t(\omega)(s)=\omega (t+s) - \omega(t)\,,\]
for all $t$, $s$ in $\RR_+$ and all $\o$.
\smallskip

The convergence of $\br{\t}_t^{(h)}$ to $\t_t$ according to the natural metric $D$ on $\O$ for all $t$ is shown in the next theorem.

\begin{theo}
Let $\omega$ be a function in $\Omega$. For all $t \in \mathbb{R}_+$,
$$
\lim_{h \rightarrow 0}\,D(\theta_t(\omega), \br{\theta}^{(h)}_t(\omega)) = 0\,.
$$
\end{theo}

\begin{proof}
As for the proof of the lemma 1, this result is also based on the convergence of piecewise linear functions to continuous one on compact sets. However, note that two linear interpolations constitute the definition of $\br{\t}_t^{(h)}$ instead of one: one is due to the injection $\f_I^{(h)}$ and the other one is due to the construction of the continuous dynamics on $\RR^m \times \O$.

\smallskip
For all $\omega \in \Omega$ and for all $t$ and $s$ in $\mathbb{R}_+$,
we start by computing $\theta_t(\omega)(s) - \br{\theta}^{(h)}_t(\omega)(s)$.
By definition of $\br{T}^{(h)}_t$, the point $\br{\theta}^{(h)}_t(\omega)(s)$ is obtained by linear interpolation between $\f_I^{(h)} \circ (\theta^{(h)})^{\lfloor t/h\rfloor} \circ \f_P^{(h)}(\o) (s)$ and $\f_I^{(h)} \circ (\theta^{(h)})^{(\lfloor t/h\rfloor +1)} \circ \f_P^{(h)} (\o) (s)$.
Therefore, let us compute these two values. We have
\begin{multline*}
\f_I^{(h)} \circ (\theta^{(h)})^{\lfloor t/h\rfloor} \circ \f_P^{(h)} (\omega) (s) =
\omega((\lfloor t/h\rfloor +\lfloor s/h\rfloor)h) - \omega(\lfloor t/h\rfloor h)+\hfill\\ 
\hfill + \frac{s-\lfloor s/h\rfloor h}{h}\Big{\{}\omega((\lfloor t/h \rfloor + \lfloor s/h \rfloor +1)h) - \omega((\lfloor t/h \rfloor+\lfloor s/h \rfloor)h)\Big{\}}\,.
\end{multline*}
On the other hand,
\begin{multline*}
\f_I^{(h)} \circ (\theta^{(h)})^{(\lfloor t/h\rfloor +1)} \circ \f_P^{(h)} (\omega) (s) =
\omega((\lfloor t/h\rfloor +\lfloor s/h\rfloor+1)h) - \omega((\lfloor t/h\rfloor+1) h)+\hfill\\ 
\hfill+ \frac{s-\lfloor s/h\rfloor h}{h}\Big{\{}\omega((\lfloor t/h \rfloor + \lfloor s/h \rfloor +2)h) - \omega((\lfloor t/h \rfloor+\lfloor s/h \rfloor +1)h)\Big{\}}\,.
\end{multline*}
Since $\br{\theta}^{(h)}_t(\omega)(s)$ is a barycenter between these two points whose coefficients are given by the linear interpolation on $t$, then
\begin{align*}
\br{\theta}^{(h)}_t(\omega)(s) &=
\frac{(\lfloor t/h\rfloor+1)h - t}{h}(\f_I^{(h)} \circ (\theta^{(h)})^{\lfloor t/h\rfloor} \circ \f_P^{(h)})(\omega) (s)+\\
&\ + \frac{t-\lfloor t/h\rfloor h}{h} (\f_I^{(h)} \circ (\theta^{(h)})^{(\lfloor t/h\rfloor +1)} \circ \f_P^{(h)})(\omega) (s)\\
&=\frac{(\lfloor t/h\rfloor+1)h - t}{h} \Big{\{} \omega((\lfloor t/h\rfloor+\lfloor s/h\rfloor)h)+\\ 
&\ +\frac{s-\lfloor s/h\rfloor h}{h}\Big{[}\omega((\lfloor t/h \rfloor + \lfloor s/h \rfloor +1)h) - \omega((\lfloor  t/h \rfloor+\lfloor  s/h \rfloor)h)\Big{]}\Big{\}}+\\ 
&\ + \frac{t-\lfloor t/h\rfloor h}{h}\Big{\{}\omega((\lfloor t/h\rfloor +\lfloor s/h\rfloor+1)h)+\\ 
&\ +\frac{s-\lfloor s/h\rfloor h}{h}\Big{[}\omega((\lfloor t/h \rfloor+ \lfloor s/h \rfloor +2)h) - \omega((\lfloor t/h \rfloor+\lfloor s/h \rfloor+1)h)\Big{]}\Big{\}}+\\
&\ - \frac{( \lfloor t/h\rfloor +1)h - t}{h}\,\omega(\lfloor t/h\rfloor h)- \frac{t-\lfloor t/h\rfloor h}{h} \,\omega((\lfloor t/h\rfloor+1) h)\,.
\end{align*}
As seen in the proof of the lemma 1, the term 
$$- \frac{(\lfloor t/h\rfloor+1)h - t}{h}\,\omega(\lfloor t/h\rfloor h)- \frac{t-\lfloor t/h\rfloor h}{h}\,\omega((\lfloor t/h\rfloor+1) h)$$ 
tends to $-\omega(t)$ when $h$ goes to $0$.

We just have to prove now that the other terms converge to $\omega(t+s)$.
As the function $\omega$ is continuous, we have
\[ \lim_{h \rightarrow 0}\left|\omega((\lfloor t/h \rfloor + \lfloor s/h \rfloor +1)h) - \omega((\lfloor t/h \rfloor+\lfloor s/h \rfloor)h)\right| =0\,.\]
For the same reason, 
\[\lim_{h \rightarrow 0}\left|\omega((\lfloor t/h \rfloor+ \lfloor s/h \rfloor +2)h) - \omega((\lfloor t/h \rfloor+\lfloor s/h \rfloor+1)h)\right|=0\,.\] 
On the other hand,
\[0 \leq \frac{(\left[t/h\right]+1)h - t}{h} \leq 1\,,\hspace{5mm} 0 \leq \frac{t-\left[t/h\right]h}{h} \leq 1\,, \]
and, obviously,
\[  0 \leq \frac{(\left[s/h\right]+1)h - s}{h} \leq 1,\hspace{5mm} 0 \leq \frac{s-\left[s/h\right]h}{h} \leq 1\,, \]
Therefore,
\begin{multline*}
\lim_{h \rightarrow 0}\, \left|\frac{(\lfloor t/h\rfloor+1)h - t}{h} \right|\,\left|\frac{s-\lfloor s/h\rfloor h}{h}\right|\times\hfill\\
\hfill\times\left|\omega((\lfloor t/h \rfloor+ \lfloor s/h \rfloor +2)h) - \omega((\lfloor t/h \rfloor+\lfloor s/h \rfloor+1)h)\right| =0\,,
\end{multline*}
and
\begin{multline*}
\lim_{h \rightarrow 0} \, \left|\frac{t-\lfloor t/h\rfloor h}{h} \right|\,\left|\frac{s-\lfloor s/h\rfloor h}{h}\right| \times \hfill \\
\hfill \times \left|\omega((\lfloor t/h \rfloor + \lfloor s/h \rfloor +1)h) - \omega((\lfloor  t/h \rfloor+\lfloor  s/h \rfloor)h)\right| =0\,.
\end{multline*}
For the same reasons as previously, the remaining terms 
$$\frac{(\lfloor t/h\rfloor+1)h - t}{h} \omega((\lfloor t/h\rfloor+\lfloor s/h\rfloor)h)+\frac{t-\lfloor t/h\rfloor h}{h}\omega((\lfloor t/h\rfloor +\lfloor s/h\rfloor+1)h)$$ 
tend to $\o(t+s)$.

As a conclusion, for all $t$, $s$ in $\RR_+$, we have
 \[\lim_{h \rightarrow 0} \, \left| \theta_t(\omega)(s) - \br{\theta}^{(h)}_t(\omega)(s) \right| = 0\,.\]
Since the interval $[0,n]$  is compact, by uniform continuity of these applications
\[\lim_{h \rightarrow 0}\,\underset{0\leq s \leq n}{\sup} \left| \theta_t(\omega)(s) - \br{\theta}^{(h)}_t(\omega) (s) \right| =0\,.\]
Finally, by Lebesgue's Theorem,
\[\lim_{h \rightarrow 0}\,D (\theta_t(\omega), \br{\theta}^{(h)}_t(\omega)) = 0\,.\]
The theorem is proved. \end{proof}

Eventually, the application $\br{\theta}^{(h)}_t$ converges to the shift $\t_t$ when $h$ goes to $0$ for all $t$ on $\O$ whatever be the continuous time dynamical system as long as the noise of the stochastic differential equation is a $d$-dimensional Brownian motion whose canonical space is $\O$.

\subsection{$L^{p}$ and Almost-Sure Convergence}

After having studied the convergence of the shift, we want now to give conditions on $U^{(h)}$ for the  $L^{p}$ and almost sure convergence, on every time interval $[0,\T]$, of the process $\br{X}^h_t$, the first component of $\br{T}^h_t$, to the solution $X_t$ of a SDE.

\smallskip
As the process $\br{X}^h_t$ is just a linearly interpolated Markov chain, this convergence boils down to the convergence of some schemes of stochastic numerical analysis. Thus, our result belongs to a more general problem, that we shall apply later on to the process $\br{X}^h_t$ in Theorem 5.3.

\smallskip
Consider the solution $X_t$ in $\RR^m$ starting in $X_0$ of the SDE \eqref{eq:A}
\[dX_t = b(X_t) dt + \sigma(X_t) dW_t\,,\]
where $(W_t)$ is a $d$-dimensional Brownian motion, and where the applications $b$ and $\s$, respectively from $\RR^m$ to $\RR^m$ and from $\RR^m$ to $\mathcal{M}_{m,d}(\mathbb{R})$, are Lipschitz. Recall that this assumption is required for the existence and the uniqueness of the solution of the SDE on every time interval $[0, \T]$ and for all initial conditions.

Let $(X^h_{nh})_{n \in \NN}$ be a Markov chain for a time step $h$ whose evolution is given by 
\[X^h_{(n+1)h}= X_{nh}^h+ \sigma(X_{nh}^h)(W_{(n+1)h}- W_{nh}) + h b(X_{nh}^h) + h \eta^{(h)}(X_{nh}^h,W_{(n+1)h}- W_{nh}),\]
with $X^h_0 = x_0$ and where $\eta^{(h)}$ is a measurable application.

Note that the form  of the equation above is identical to the ones previously seen in the physical examples.
Also note that, without the term $\eta^{(h)}$, the scheme above is the usual stochastic Euler one. Hence our context is more general than the usual Euler scheme for the discrete-time approximation of SDE. We have to adapt convergence theorem to this situation.

\bigskip
The Markov chain $(X_{nh}^{h})$ is now linearly interpolated to obtain a continuous time process $(X_t^{h})_{t \in \RR_+}$ defined by
\begin{multline*}
 X_t^h = X_{\lfloor t/h\rfloor h}^h +\dfrac{t-\lfloor t/h\rfloor h}{h} \big{\{} X_{(\lfloor t/h\rfloor +1)h}^h-X_{\lfloor t/h\rfloor h}^h\big{\}}\\
\phantom{X_t^h\ \ \ }= X_{\lfloor t/h\rfloor h}^h +\dfrac{t-\lfloor t/h\rfloor h}{h} \big{\{}\sigma(X_{\lfloor t/h\rfloor h}^h)(W_{(\lfloor t/h\rfloor+1)h}- W_{\lfloor t/h\rfloor h})+\hfill \\
\hfill +h b(X_{\lfloor t/h\rfloor h}^h)+h \eta^{(h)}(X_{\lfloor t/h\rfloor h}^h,W_{(\lfloor t/h\rfloor+1)h}- W_{\lfloor t/h\rfloor h})\big{\}}\,,
\end{multline*}
for all $t$, and with $X^h_0=x_0$.

For the convergence of the process $X^h_t$, we make one more assumption to control the term in  $\eta^{(h)}$.
\begin{enumerate}
\item[(H4)] There exist $\alpha \in ]0, + \infty]$ and $K_2$ such that
\[\left|\eta^{(h)}(x,y) \right| \leq K_2 (h^{\alpha}\left|x\right| + \left|y\right|)\,.\]
\end{enumerate}
We start with the main result on the convergence of processes in the case of globally Lipschitz functions.
\begin{theo}\label{globally}
For all $\T>0$ and for all $p>2$, under Assumptions (H1) and (H4), the process $(X_t^h)$ converges in $L^{p}$ to the solution $(X_t)$ of the SDE \eqref{eq:A} on $\left[0,\T\right]$.

More precisely, for $h$ small enough and $p=2q$ with $q>1$,
\[\mathbb{E}\Big{[}\big( \sup_{t \in \left[0,\T\right]} \left| X_t - X^h_t \right|\big)^{2q}\Big{]} \leq C (h^{2q\alpha} + h^{q-1} (-\log h)^q)\,\]
where $C$ is a constant depending only on $\T$, $q$, $K_0$ and $K_2$.

Moreover, if $q$ is such that $q>2$ and $2q\alpha >1$, then the convergence is almost sure on $\left[0,\T\right]$.
\end{theo}

In order to prove this theorem, we use the same strategy as Faure considered in his PhD thesis (\cite{F}) for the convergence of the explicit Euler scheme, i.e. without the term $h\eta^{(h)}$.
Before hands we need two long and technical lemmas \ref{le2} and \ref{le4} and a property on solution of stochastic differential equation, Lemma \ref{le3}. These lemmas shall be applied in the proof in the case of locally Lipschitz and linearly bounded applications because the linear growth condition shall be the key property.

\smallskip
The first lemma gives an inequality on the $L^{p}$-norm of some stochastic processes.
In order to obtain this result, the definition of a \emph{$L^{p}$-continuous process} is introduced.

A process $Y_t$ is \emph{$L^{p}$-continuous} if the function $t \longmapsto \mathbb{E}\left[\left|Y_t\right|^{p}\right]$ is continuous.

\begin{lemm}\label{le2}
Let $Y_t$ be a process defined by
$Y_t= Y_0 + \int_0^t A_s\, dW_s + \int_0^t B_s\, ds\,,$
where $A_s$ and $B_s$ are $L^{2p}$-continuous, and $\EE \left[\left|Y_0\right|^{p}\right] < \infty$.
Then $Y_t$ is $L^{p}$-continuous.

Moreover,
\begin{align}
\mathbb{E}\left[\left|Y_t\right|^{p}\right] \leq \mathbb{E}\left[\left|Y_0\right|^{p}\right] + C \int_0^t \mathbb{E}\left[\left|Y_s\right|^{p} + \|A_s\|^{p}  + \left|B_s\right|^{p}\right] \,ds\,.\label{22}
\end{align}
\end{lemm}

\begin{proof}
First, by the convexity of the function $x \longrightarrow x^{q}$, we have for all $x$, $y$, and $z$ non-negative reals 
\[(x+y+z)^{2q} \leq C_{2q} (x^{2q}+ y^{2q}+z^{2q})\,.\]
Hence,
\[\EE\left[\left|Y_t\right|^{2q}\right] \leq C_{2q} \Big{(}\EE\left[\left|Y_0\right|^{2q}\right] + \EE\Big{[}\Big{|}\int_0^t A_s\,dW_s \Big{|}^{2q}\,\Big{]} + \EE\Big{[}\Big| \int_0^t B_s\,ds\, \Big| ^{2q}\,\Big{]}\Big{)} \,.\]
By H\"older's inequality we claim that
\begin{equation}\label{holder}
\EE\Big[\Big|\int_0^t B_s\,ds\, \Big|^{2q}\,\Big] \leq C_0\, t^{2q-1} \int_0^t \EE\left[\l B_s\r^{2q}\,\right]\,ds\,.
\end{equation}
Indeed, first note that 
\[\Big| \int_0^t B_s\,ds\, \Big|^{2q}= \Big( \sum_{i=1}^m \Big|\int_0^t B^i_s\,ds\, \Big|^{2}\,\Big)^q\,.\]
Component by component, we have by H\"older's inequality,
$$
\EE\Big[\Big|\int_0^t B_s^i\,ds\, \Big|^{2q}\,\Big] \leq  t^{2q-1}\EE\Big[\int_0^t \left|B_s^i \right|^{2q} \,ds\Big]\,.$$
By the convexity of the application $t\longmapsto t^q$ we have the announced inequality (\ref{holder}).
For the second term, there exists Burkholder inequality based on It\^o's formula which gives a similar bound:
\[\EE\Big[\Big|\int_0^t A_s\,dW_s \,\Big|^{2q}\,\Big] \leq C_1 t^{q-1}  \int_0^t \EE\left[\| A_s\|^{2q}\right]\,ds\,.\]
Finally, we have obtained the following bound
\[\EE \left[\left|Y_t\right|^{2q}\right] \leq C_{2q} \Big{(}\EE\left[\left|Y_0\right|^{2q}\right] + C_0 t^{2q-1} \int_0^t \EE\left[\l B_s\r^{2q}\right]\,ds + C_1 t^{q-1}  \int_0^t \EE\left[\| A_s\|^{2q}\right]\,ds \Big{)}\,.\]
Now for the $L^{p}$-continuity of this process, It\^o's formula is applied between two times $s, t$ with $s \leq t$ to the process $(Y_t)$.
Indeed, since the application $x \longmapsto \left|x\right|^{2q}$ is twice differentiable for $q\geq 1$, we get
\begin{align*}
\left|Y_t\right|^{2q} =\left|Y_s\right|^{2q} &+ \sum_{i=1}^m \int_s^t 2q \left|Y_u\right|^{2q-2} Y_u^i \,dY_u^i \ +\\
& + \frac{1}{2} \sum_{i\neq j} \int_s^t 2q(2q-2)\left|Y_u\right|^{2q-4}Y_u^i Y_u^j \,d\langle Y^i,Y^j \rangle_u \ +\\
&+ \frac{1}{2} \sum_{i} \int_s^t 2q(2q-2) (Y_u^i)^2 \left|Y_u\right|^{2q-4} + 2q\left|Y_u\right|^{2q-2}\, d\langle Y^i,Y^i \rangle_u
\end{align*}
where $Y_u=(Y_u^i)_{i=1,\cdots,m}$\,.

From the definition of $Y_u$, 
$$dY_u^i= \sum_{j=1}^d A_u^{i,j}\,dW^j_u + B_u^i\, du\,,$$
we have
$$d\langle Y^i,Y^j \rangle_u = \sum_{k=1}^d \sum_{l=1}^d A_u^{i,k} A_u^{j,l}\,d\langle W^k_, W^l \rangle_u = \sum_{k=1}^d A_u^{i,k} A_u^{j,k}\, du\,.$$
Hence we get
\begin{align*}
\left|Y_t\right|^{2q} = \left|Y_s\right|^{2q} &+ 2q \sum_{i=1}^m \int_s^t  \left|Y_u\right|^{2q-2} Y_u^i B_u^i\, du +2q \sum_{i=1}^m \sum_{j=1}^d \int_s^t  \left| Y_u\right|^{2q-2} Y_u^i A_u^{i,j}\,dW^j_u +\\
&+ 2q(q-1) \sum_{i,j} \sum_{k=1}^d \int_s^t \left|Y_u\right|^{2q-4}Y_u^i Y_u^j A_u^{i,k} A_u^{j,k} \,du +\\
&+ p \sum_{i=1}^m \sum_{k=1}^d \int_s^t \left|Y_u\right|^{2q-2}( A_u^{i,k})^2 \,du\,.
\end{align*}
Taking the expectation, we obtain
\begin{align*}
\mathbb{E}\left[\left|Y_t\right|^{2q}\right] =& \mathbb{E}\left[\left|Y_s\right|^{2q}\right] + 2q \sum_{i=1}^m \int_s^t 
\mathbb{E}\left[\left|Y_u\right|^{2q-2} Y_u^i B_u^i\right]\,du+ \nonumber \\
 &\qquad+ 2q(q-1) \sum_{i,j} \sum_{k=1}^d \int_s^t \EE\left[\left|Y_u\right|^{2q-4}Y_u^i Y_u^j A_u^{i,k} A_u^{j,k}\right] \,du + \nonumber\\
&\qquad+ p \sum_{i=1}^m \sum_{k=1}^d \int_s^t \EE\left[\left|Y_u\right|^{2q-2}( A_u^{i,k})^2 \right]\,du\,.
\end{align*}
Hence,
\begin{multline}
\left|\mathbb{E}\left[\left|Y_t\right|^{2q}\right] - \mathbb{E}\left[\left|Y_s\right|^{2q}\right]\right|\leq  C_1 \int_s^t 
\mathbb{E}\left[\left|Y_u\right|^{2q-1} \left| B_u\right|\right]\,du + \hfill \\
\hfill + C_2 \int_s^t \EE\left[\left|Y_u\right|^{2q-2} \|A_u\|^2\right] \,du\,.\label{44}
\end{multline}
Consider the first term in Inequality (\ref{44}). By H\"older's inequality we have
\[\mathbb{E}\left[\left|Y_u\right|^{2q-1} \left| B_u\right|\right] \leq \EE\left[\left|Y_u\right|^{2q}\right]^{(2q-1)/2q}\ \EE\left[\left| B_u\right|^{2q}\right]^{1/2q}\,.\]
In the same way for the second term in (\ref{44}) we get
 \[\mathbb{E}\left[\left|Y_u\right|^{2q-2} \|A_u\|^2\right] \leq \EE\left[\left|Y_u\right|^{2q}\right]^{(q-1)/q}\ \EE\left[\|A_u\|^{2q}\right]^{1/q}\,.\]
Therefore, 
 \begin{multline*}
\left|\mathbb{E}\left[\left|Y_t\right|^{2q}\right] - \mathbb{E}\left[\left|Y_s\right|^{2q}\right]\right|\leq  C_1 \int_s^t 
\EE\left[\left|Y_u\right|^{2q}\right]^{(2q-1)/2q}\ \EE\left[\left| B_u\right|^{2q}]\right]^{1/2q}\,du + \hfill\\
\hfill + C_2 \int_s^t \EE\left[\left|Y_u\right|^{2q}\right]^{(q-1)/q}\ \EE\left[\|A_u\|^{2q}\right]^{1/q} \,du\,.
\end{multline*}
Since $\EE\left[\left|Y_u\right|^{2q}\right]$ is bounded for all $s\leq u \leq t$ and by the $L^{p}$-continuity of $A_u$ and $B_u$, one can conclude that the process $Y_t$ is $L^{p}$-continuous.

Let us proceed now with the proof of Inequality (\ref{22}).
By It\^o's formula between $t$ and $0$ we have
\begin{align*}
\left|Y_t\right|^{2q} =\left|Y_0\right|^{2q} &+ \sum_{i=1}^m \int_0^t 2q \left|Y_s\right|^{2q-2} Y_s^i dY_s^i + \\
&+\frac{1}{2} \sum_{i\neq j} \int_0^t 2q(2q-2)\left|Y_s\right|^{2q-4}Y_s^i Y_s^j d\langle Y^i,Y^j \rangle_s +\\
&+ \frac{1}{2} \sum_{i} \int_0^t 2q(2q-2) (Y_s^i)^2 \left|Y_s\right|^{2q-4} + 2q\left|Y_s\right|^{2q-2} d\langle Y^i,Y^i \rangle_s\,.
\end{align*}
Hence, taking the expectation we get
\begin{align}
\mathbb{E}\left[\left|Y_t\right|^{2q}\right] = \mathbb{E}\left[\left|Y_0\right|^{2q}\right] &+ 2q \sum_{i=1}^m \int_0^t 
\mathbb{E}\big[\left|Y_s^i\right|^{2q-2} Y_s^i B_s^i\big]\,ds +\nonumber \\
 &+ 2q(q-1) \sum_{i,j} \sum_{k=1}^d \int_0^t \EE\left[\left|Y_s\right|^{2q-4}Y_s^i Y_s^j A_s^{i,k} A_s^{j,k}\right] \,ds +\nonumber\\
&+ q \sum_{i=1}^m \sum_{k=1}^d \int_0^t \EE\left[\left|Y_s\right|^{2q-2}( A_s^{i,k})^2 \right]\,ds\,.\label{55}
\end{align}
Let us start with the last term of the (\ref{55}).
First recall that
\[\sum_{i=1}^m \sum_{k=1}^d ( A_s^{i,k})^2 = \|A_s \|^2\,.\]
Note that, for all $k$ in $\left[\!\left| 0 , 2q \right|\!\right]$ and $x,y$ in $\mathbb{R}$, we claim that $x^{2q-k}y^{k} \leq x^{2q} + y^{2q}$. Indeed we have  $Y^k \leq 1+ Y^{2q}$, for all $Y$ in $\mathbb{R}$ and we apply it to $Y= y/x$ (if $x=0$, the inequality is  clearly true).
Therefore we get
\[\sum_{i=1}^m \sum_{k=1}^d \left|Y_s\right|^{2q-2} ( A_s^{i,k})^2 \leq \left|Y_s\right|^{2q} + \|A_s \|^{2q}\,.\]
Now consider the first term in (\ref{55}). If the scalar product on $\RR^m$ is denoted by $(\,,\,)$, then
\[
\sum_{i=1}^m \left|Y_s\right|^{2q-2} Y_s^i B_s^i = \left|Y_s\right|^{2q-2} (Y_s,B_s) \leq \left|Y_s\right|^{2q-1} \left|B_s \right|\,.\]
From the previous inequality, one obtains 
\[\sum_{i=1}^m \left|Y_s\right|^{2q-2} Y_s^i B_s^i \leq  \left|Y_s\right|^{2q} + \left|B_s\right|^{2q}\,.\]
Let us  consider now the last term of (\ref{55}).
Note that,
\[ \left|Y_s\right|^{2q-4} Y_s^i Y_s^j A_s^{i,k} A_s^{j,k} \leq \left|Y_s\right|^{2q-2} \|A_s \|^{2}\,.\]
Hence,
\[\sum_{i,j} \sum_{k=1}^d  \left|Y_s\right|^{2q-4} Y_s^i Y_s^j A_s^{i,k} A_s^{j,k}  \leq C_3( \left|Y_s\right|^{2q} +  \|A_s\|^{2q})\,.\]
In conclusion, there exists a constant $C$ such that
\[\mathbb{E}\left[\left|Y_t\right|^{2q}\right] \leq \mathbb{E}\left[\left|Y_0\right|^{2q}\right] + C \int_0^t \mathbb{E}\left[\left|Y_s\right|^{2q} + \|A_s\|^{2q}  + \left|B_s\right|^{2q}\right]\,ds\,.\]
The lemma is proved.
\end{proof}

The next lemma (proved in \cite{KP}) gives some regularities of trajectories of solutions $X_t$ of the SDE \eqref{eq:A}.

\begin{lemm}\label{le3}

Let $(X_t)$ be the solution of \eqref{eq:A} for all $t \in \left[0,\T\right]$. Suppose that the maps $b$ and $\s$ are locally Lipschitz (H2) and linearly bounded (H3) . 

Then, for all $t \in \left[0,\T\right]$ and for all $q\leq 1$,
\begin{align}
\mathbb{E}\left[\left| X_{t}\right|^{2q}\right] \leq (1+ \mathbb{E}\left[\left| x_0 \right|^{2q}\right]) e^{Ct},\label{3311}
\end{align}
and, for all $t$, $s$ such that $t\geq s$,
\begin{align}
\mathbb{E}\left[\left| X_{t}-X_s\right|^{2q}\right] \leq D  (1+ \mathbb{E}\left[\left| x_0 \right|^{2q}\right]) (t-s)^q e^{C (t-s)},\label{3322}
\end{align}
where $C$ and $D$ are positive constants depending only on $\T$, $q$ and $K_1$.

Moreover, for all $\T>0$
\begin{align*}
\mathbb{E}\left[\sup_{t \in \left[ 0,\T\right]}\left| X_{t}\right|^{2q}\right] < + \infty
\end{align*}
\end{lemm}
For the convergence of the process $X^h_t$ to $X_t$ the main tool shall be Lemma \ref{le2}.
However, notice that the evolution of $X^h_t$ doesn't allow to apply this lemma because of the linear interpolation. More precisely, for all $n$, between the times $nh$ and $(n+1)h$, this process is not of the form $X_{nh}^h + \int_{nh}^t A_s\, dW_s + \int_{nh}^t B_s\, ds$. Therefore, in order to apply it, it's natural to introduce the new process $Y^h_t$ defined by
\begin{multline*}
Y^h_{t} = Y^h_{\lfloor t/h \rfloor h} + \int_{\lfloor t/h \rfloor  h}^{t} b(Y^h_{\lfloor t/h \rfloor  h}) +\eta^h( Y^h_{\lfloor t/h \rfloor  h}, W_{(\lfloor t/h \rfloor  +1)h} - W_{\lfloor t/h \rfloor  h}))\,ds + \hfill\\
 \hfill +\int_{\lfloor t/h \rfloor  h}^{t} \sigma(Y^h_{\lfloor t/h \rfloor  h}) \,dW_s\,.
 \end{multline*}
Note that $X^h_{(n +1)h}=Y^h_{(n +1)h}$ for all $n$.

\smallskip
The last lemma gives equivalent bounds as in Lemma \ref{le3} but for the processes $(X^h_t)$ and $(Y_t^h)$.

\begin{lemm}\label{le4}
Let $(X^h_t)$ and $(Y^h_t)$ be the processes defined above.
Suppose that the maps $b$ and $\s$ are locally Lipschitz and linearly bounded. Moreover, suppose that Hypothesis (H4) are satisfied. Then, for all $t \in \left[0,\T\right]$ and all $q \geq 1$,
\begin{align*}
\mathbb{E}\big[\left|X^h_t\right|^{2q}\big] \leq C_0\big(1+\mathbb{E}\big[\left|X^h_0\right|^{2q}\big]\big)e^{C_1t}\,,
\end{align*}
and,
\begin{align}
\mathbb{E}\big[\left|Y^h_t\right|^{2q}\big] \leq C_0\big(1+\mathbb{E}\big[\left|X^h_0\right|^{2q}\big]\big)e^{C_1t}\,.\label{411}
\end{align}
Moreover, for all $t$ and for $h$ small enough,
\begin{align*}
\mathbb{E}\big[\left|X^h_t-X^h_{\lfloor t/h\rfloor h}\right|^{2q}\big] \leq C_2 (h^{2q} + h^q(-\log h)^q)\,,
\end{align*}
and,
\begin{align}
\mathbb{E}\big[\left|Y^h_t-Y^h_{\lfloor t/h\rfloor h}\right|^{2q}\big] \leq C_2 (h^{2q} + h^q(-\log h)^q)\,.\label{422}
\end{align}
On the other hand, for all $h\leq h_0$,
\begin{align*}
\EE\big[ (\sup_{ t \leq \T} \l Y^h_t \r)^{2q} \big] \leq C_3 \big(1+\mathbb{E}\big[\left|X^h_0\right|^{2q}\big]\big)\,,
\end{align*}
where $(C_i)$ are constants depending only on $\T$, $q$, $K_1$ and $K_2$.
\end{lemm}

\begin{proof} Before the proof, we want to note that Assumption (H1) implies the fact that the functions $b$ and $\s$ are linearly bounded. More precisely, there exists $K_1\geq 0$ such that
\[\left|b(x)\right| \leq K_1(1 + \left|x\right|)\qquad\mbox{and}\qquad
\|\sigma(x)\|\leq K_1(1 + \left|x\right|)\,.\]
This linear growth property shall be often used in the following proofs.

\smallskip
The proof of this lemma is achieved in several steps.

The first step is to bound $\mathbb{E}\big[\left|X^h_t\right|^{2q}\big]$ according to $\mathbb{E}\big[\big|X^h_{\lfloor t/h\rfloor h}\big|^{2q}\,\big]$. From the definition of $X^h_t$, we have
\begin{align*}
\left|X^h_t\right|^{2q} \leq C_0 \big(\left|X^h_{\lfloor t/h\rfloor h}\right|^{2q} + &h^{2q} \left|b(X^h_{\lfloor t/h\rfloor h})\right|^{2q} +\\
&\qquad+ \|\sigma(X^h_{\lfloor t/h\rfloor h})\|^{2q}\left|W_{(\lfloor t/h\rfloor +1)h} - W_{\lfloor t/h\rfloor h}\right|^{2q}\\ 
&\qquad+ h^{2q} \left|\eta^h( X^h_{\lfloor t/h\rfloor h}, W_{(\lfloor t/h\rfloor +1)h} - W_{\lfloor t/h\rfloor h})\right|^{2q}\big)\,.
\end{align*}
The process $X^h_t$ at time $\lfloor t/h\rfloor h$ is independent of $W_{(\lfloor t/h\rfloor +1)h} - W_{\lfloor t/h\rfloor h}$. Therefore, by the linear growth property, we get
\begin{multline}
\mathbb{E}\big[\left|X^h_t\right|^{2q}\big] \leq C_0 \Big{(} \mathbb{E}\big[\left|X^h_{\lfloor t/h\rfloor h}\right|^{2q}\big] +C_pK_1^{2q} h^{2q} (1+\mathbb{E}\big[\left|X^h_{\lfloor t/h\rfloor h}\right|^{2q}\big])+\hfill  \\
\hfill + C_pK_1^{2q} (1+\mathbb{E}\big[\left|X^h_{\lfloor t/h\rfloor h}\right|^{2q}\big]) \mathbb{E}\big[\left|W_{(\lfloor t/h\rfloor+1)h} - W_{\lfloor t/h\rfloor h}\right|^{2q}\big] +\\ 
\hfill + h^{2q} \mathbb{E}\big[\left|\eta^h( X^h_{\lfloor t/h\rfloor h}, W_{(\lfloor t/h\rfloor+1)h} - W_{\lfloor t/h\rfloor h})\right|^{2q}\big]\Big{)}\,.\label{433}
\end{multline}
In the same way as above, the following bound can be found for  $(Y_t^h)$:
\begin{multline}
\mathbb{E}\big[\left|Y^h_t\right|^{2q}\big] \leq C_0 \Big{(} \mathbb{E}\big[\left|Y^h_{\lfloor t/h\rfloor h}\right|^{2q}\big] +C_qK_1^{2q} h^{2q} (1+\mathbb{E}\big[\left|Y^h_{\lfloor t/h\rfloor h}\right|^{2q}\big]) + \hfill \\
\hfill + C_pK_1^{2q} (1+\mathbb{E}\big[\left|Y^h_{\lfloor t/h\rfloor h}\right|^{2q}\big]) \mathbb{E}\big[\left|W_{t} - W_{\lfloor t/h\rfloor h}\right|^{2q}\big]+ \\
\hfill+ h^{2q} \mathbb{E}\big[\left|\eta^h( Y^h_{\lfloor t/h\rfloor h}, W_{(\lfloor t/h\rfloor+1)h} - W_{\lfloor t/h\rfloor h})\right|^{2q}\big]\Big{)}\,.\label{444}
\end{multline}
The next step is now to bound $\mathbb{E}\big[\left|W_{(\lfloor t/h\rfloor +1)h} - W_{\lfloor t/h\rfloor h}\right|^{2q}\big]$ and\\ $\mathbb{E}\big[\left|W_{t} - W_{\lfloor t/h\rfloor h}\right|^{2q}\big]$.
By definition of the norm, one can find an upper bound by regarding the supremum of $d$ $1$-dimensional standard Brownian motions $(B^k_t)_{k=1 \cdots d}$ on the time interval $\left[0,h\right]$.
Thus, consider the process $M_h$ defined by
\[M_h=\max_{k \in \left[\left|1,d\right|\right]} \sup_{t \in \left[0,h\right]} \left|B^k_t \right|\,.\]
The aim is to find a bound on $\mathbb{E}\left[M_h^{2q}\right]$. Firstly, note that, 
\[\mathbb{E}\left[M_h^{2q}\right] \leq \mathbb{E}\Big[M_h^{2q}\ \mathds{1}_{\left\{M_h>2\sqrt{h(-\log h)}\right\}}\Big]+ C_1 h^q(-\log h)^p\,.\]
But,
\[\mathbb{E}\Big[M_h^{2q}\ \mathds{1}_{\left\{M_h>2\sqrt{h(-\log h)}\right\}}\Big] \leq \sum_{k=1}^{d} \mathbb{E}\Big[(\sup_{t \in \left[0,h\right]} \left|B^k_t \right|)^{2q} \ \mathds{1}_{\big\{\underset{t \in \left[0,h\right]}{\sup} \left|B^k_t \right|>2\sqrt{h(-\log h)}\big\}}\Big]\,.\]
By using the reflexion principle we get
\begin{align*}
\mathbb{E}\Big[M_h^{2q}\ \mathds{1}_{\left\{M_h>2\sqrt{h(-\log h)}\right\}}\Big] &\leq 2\sum_{k=1}^{d} \mathbb{E}\Big[ (\sup_{t \in \left[0,h\right]} B^k_t)^{2q}\ \mathds{1}_{\big\{\underset{t \in \left[0,h\right]}{\sup} B^k_t >2\sqrt{h(-\log h)}\big\}}\Big]\\
&\leq C_2\int_{x\geq2\sqrt{h(-\log h)}} x^{2q} g(x)\,dx\,,
\end{align*}
where $g(x)= 2 \ e^{-x^2/2h}/ \sqrt{2\pi h}\,.$

Hence,
\[\mathbb{E}\Big[M_h^{2q}\ \mathds{1}_{\left\{M_h>2\sqrt{h(-\log h)}\right\}}\Big] \leq C_3 h^{q-1/2} \int_{u\geq2\sqrt{-\log h}} u^{2q}e^{-u^2/2}\,du\,.\]
Let us compute now this integral.
By integration by parts, we have
\begin{align*}
I_{q}=\int_{u\geq2\sqrt{-\log h}} u^{2q}e^{-u^2/2}\,du =& \left[u^{2q-1}(-e^{-u^2/2})\right]^{\infty}_{2\sqrt{-\log h}}\\ &+ (2q-1)\int_{u\geq2\sqrt{-\log h}} u^{2q-2}e^{-u^2/2}dx\\
=&(2\sqrt{-\log h})^{2q-1} e^{-(2\sqrt{-\log h})^2/2} + (2q-1) I_{q-1}\\
=&(2\sqrt{-\log h})^{2q-1} h^2 +(2q-1) I_{q-1}\,.
\end{align*}
By recurrence, $I_q \leq C_4 (\sqrt{-\log h})^{2q-1} h^2$.
Therefore, for $h$ small enough, 
\[\mathbb{E}\left[M_h^{2q}\right] \leq C_5 h^q(-\log h)^q\,.\]
Thus, we obtain
\[\mathbb{E}\big[\left|W_{(\lfloor t/h\rfloor +1)h} - W_{\lfloor t/h\rfloor h}\right|^{2q}\big] \leq C_6 h^q(-\log h)^q\,,\]
and
\[\mathbb{E}\big[\left|W_{t} - W_{\lfloor t/h\rfloor h}\right|^{2q}\big] \leq C_6 h^q(-\log h)^q\,.\]
Then, from Assumption (H3),
\begin{multline*}
\mathbb{E}\big[\left|\eta^h( X^h_{\lfloor t/h\rfloor h}, W_{(\lfloor t/h\rfloor +1)h} - W_{\lfloor t/h\rfloor h})\right|^{2q}\big] \leq C_7 (h^{2q\alpha} \mathbb{E}\big[\left|X^h_{\lfloor t/h\rfloor h}\right|^{2q}\big] +\hfill \\
\hfill+ \mathbb{E}\big[\left|W_{(\lfloor t/h\rfloor +1)h} - W_{\lfloor t/h\rfloor h}\right|^{2q}\big])\,
\end{multline*}
and with the previous inequality on $\mathbb{E}\big[\left|W_{(\lfloor t/h\rfloor +1)h} - W_{\lfloor t/h\rfloor h}\right|^{2q}\big]$,
\begin{multline}
\mathbb{E}\big[\left|\eta^h( X^h_{\lfloor t/h\rfloor h}, W_{(\lfloor t/h\rfloor +1)h} - W_{\lfloor t/h\rfloor h})\right|^{2q}\big] \leq C_7 (h^{2q\alpha}  \mathbb{E}\big[\left|X^h_{\lfloor t/h\rfloor h}\right|^{2q}\big] +\hfill \\
\hfill+ C_6 h^q(-\log h)^q)\,.\label{466}
\end{multline}
Hence, for $h$ small enough, Inequilties (\ref{433}) and (\ref{444}) becomes
\begin{align}
\mathbb{E}\big[\left|X^h_t\right|^{2q}\big] \leq& C_8 \Big{(}(1 +  h^{2q}+  h^q(-\log h)^q)\ \mathbb{E}\big[\left|X^h_{\lfloor t/h\rfloor h}\right|^{2q}\big]+  h^{2q} +  h^q(-\log h)^q \Big{)}\,,\label{477}
\end{align}
and,
\begin{align}
\mathbb{E}\big[\left|Y^h_t\right|^{2q}\big] \leq& C_8 \Big{(}(1 +  h^{2q}+  h^q(-\log h)^q)\  \mathbb{E}\big[\left|Y^h_{\lfloor t/h\rfloor h}\right|^{2q}\big]+  h^{2q} +  h^q(-\log h)^q \Big{)}\,.\label{488}
\end{align}
These inequalities allow us to bound the expectation of the norm of the processes $X^h_t$ and $Y^h_t$ according to the time in $h\mathbb{N}^*$ just before.

The next step is to understand how the norm of this process between two successive times in $h\mathbb{N}^*$ evolves. In other words, we want to study the evolution of the norm of the Markov chain $(X_{nh}^h)=(Y_{nh}^h)$.

Recall that
\begin{multline*}
X^h_{(n+1)h}= X_{n h}^h +\sigma(X_{n h}^h)(W_{(n +1)h}- W_{n h}) + h b(X_{n}^h)+ \hfill \\
 \hfill +h \eta^{(h)}(X_{n h}^h,W_{(n+1)h}- W_{n h})  \,.
\end{multline*}
This can be also written
\begin{multline*}
X^h_{(n +1)h} = X^h_{n h} + \int_{n h}^{(n +1)h} b(X^h_{n h}) +\eta^h( X^h_{n h}, W_{(n +1)h} - W_{n h}))\,ds +\hfill \\
 \hfill +\int_{n h}^{(n +1)h} \sigma(X^h_{n h}) \,dW_s\,.
\end{multline*}
One obtains from Lemma 2 between the times $nh$ and $(n+1)h$ for the process $Y^h_{t+nh}$ that
\begin{multline*}
\mathbb{E}\big[\left|X^h_{(n +1)h}\right|^{2q}\big] \leq \mathbb{E}\big[\left|X^h_{n h}\right|^{2q}\big] + C_9 \int_{n h}^{(n +1)h} \big{(}\mathbb{E}\big[\left|Y^h_s\right|^{2q}\big]+ \mathbb{E}\big[\|\sigma(X^h_{n h})\|^{2q}\big] + \hfill \\
\hfill + \mathbb{E}\big[\left| b(X^h_{n h}) +\eta^h( X^h_{n h}, W_{(n +1)h} - W_{n h}))\right|^{2q}\big]\big{)} \,ds\,.
\end{multline*}
The value of $\mathbb{E}\big[\left|Y^h_s\right|^{2q}\big]$ can be bounded with (\ref{488}).
Hence, for $h$ sufficiently small, and from the linear growth property, it follows that
\begin{align*}
\mathbb{E}\big[\left|X^h_{(n +1)h}\right|^{2q}\big]
\leq& (1+C_{10}h)\mathbb{E}\big[\left|X^h_{n h}\right|^{2q}\big] + C_{11} h\,.
\end{align*}
Note that, the sequence $(\EE\big[\l X_{nh}^h\r^{2q}\big])_{n \in \NN}$ is subarithmetico-geometric, that is, this sequence has the following form $x_{n+1} \leq \beta x_n +\gamma$ where $\beta \geq 1$.
Thus, each $\EE\big[\l X_{nh}^h\r^{2q}\big]$ can be controlled by only $\EE\big[\l X_{0}^h\r^{2q}\big]$ and the time $n$.
Indeed, if a sequence $(x_n)$ satisfies the previous inequality, then, for all $n$, 
\[x_n \leq \beta^n x_0 + n e^{n(\beta-1)} \gamma\,.\]
Therefore, for all $t$ in $[0, \T]$,
\begin{align*}
\mathbb{E}\big[\left|X^h_{\lfloor t/h\rfloor h}\right|^{2q}\big] \leq& (1+C_{10}h)^{t/h} \mathbb{E}\big[\left|X^h_0\right|^{2q}\big] + \frac{t}{h} e^{C_{10}t} C_{11} h \\
\leq& (\mathbb{E}\big[\left|X^h_0\right|^{2q}\big]+ C_{11}t)e^{C_{10}t}\,.
\end{align*}
Hence with Inequality (\ref{477}),
\begin{align*}
\mathbb{E}\big[\left|X^h_t\right|^{2q}\big] &\leq C_{12}(\mathbb{E}\big[\left|X^h_0\right|^{2q}\big]+ C_{11}t)e^{C_{10}t}\\
&\leq C_{13}(1+ \mathbb{E}\big[\left|X^h_0\right|^{2q}\big])e^{C_{10}t}\,,
\end{align*}
and, for the same reason,
\begin{align*}
\mathbb{E}\big[\left|Y^h_t\right|^{2q}\big] \leq C_{13}(1+ \mathbb{E}\big[\left|X^h_0\right|^{2q}\big])e^{C_{10}t}\,.
\end{align*}
Let us proceed now with the proof of Inequality (\ref{422}). 
For all $t$, we get
\begin{align*}
\mathbb{E}\big[\left|X^h_t-X^h_{\lfloor t/h\rfloor h}\right|^{2q}\big] \leq C_0 \mathbb{E} \big[& h^{2q} \left|b(X^h_{\lfloor t/h\rfloor h})\right|^{2q} + \\
& \quad+ \|\sigma(X^h_{\lfloor t/h\rfloor h})\|^{2q}\left|W_{(\lfloor t/h\rfloor +1)h} - W_{\lfloor t/h\rfloor h}\right|^{2q} +\\
& \quad+ h^{2q} \left|\eta^h( X^h_{\lfloor t/h\rfloor h}, W_{(\lfloor t/h\rfloor +1)h} - W_{\lfloor t/h\rfloor h})\right|^{2q}\big]\,.
\end{align*}
Then, we have
\begin{multline*}
\mathbb{E}\big[\big|X^h_t-X^h_{\lfloor t/h\rfloor h}\big|^{2q}\big] \leq C_8 \Big{(}( h^{2q}+  h^q(-\log h)^q)
 \mathbb{E}\big[\big|X^h_{\lfloor t/h\rfloor h}\big|^{2q}\big]+\hfill \\
 \hfill+  h^{2q} +  h^q(-\log h)^q \Big{)}\,.
\end{multline*}
Finally, from the bound (\ref{411}) on $\mathbb{E}\big[\big|X^h_{\lfloor t/h\rfloor h}\big|^{2q}\big]$, we obtain
\begin{align*}
\mathbb{E}\big[\left|X^h_t-X^h_{\lfloor t/h\rfloor h}\right|^{2q}\big] \leq C_{14} ( h^{2q}+  h^q(-\log h)^q)\,
\end{align*}
and the same reasons
\begin{align*}
\mathbb{E}\big[\left|Y^h_t-Y^h_{\lfloor t/h\rfloor h}\right|^{2q}\big] \leq C_{14} ( h^{2q}+  h^q(-\log h)^q)\,.
\end{align*}
We now proceed with the proof of the last inequality.
Note that the process $Y^h_t$ can be also written by this way
\begin{align*}
Y^h_t = X_0&+\int_0^t \sum_{k=0}^{\lfloor \T/h\rfloor}  \mathds{1}_{\left] kh,(k+1)h \right]} (s) \big(b(Y_{kh}^h) + \eta^{(h)} (Y_{kh}^h,W_{(k+1)h} - W_{kh})\big)\, ds +\\
 &+ \int_0^t \sum_{k=0}^{\lfloor \T/h\rfloor}  \mathds{1}_{\left] kh,(k+1)h \right]}(s)\sigma(Y_{kh}^h)\, d W_s\,.
\end{align*}
We define the process $Z^h_t$ by
\[Z^h_t = \sup_{s \leq t} \big| Y^h_t \big|\,.\]
The aim is to find a bound on $\EE\big[(Z^h_t)^{2q}\big]$ independent of $h$.

Note that from the definition of the norm and the convexity of the application~$x \longmapsto \left| x \right|^q$, 
\begin{align*}
\mathbb{E}\big[(Z^h_t)^{2q}\big] \leq C_{15} \sum_{i=1}^m \mathbb{E}\big[\big( Z^{h,i}_t \big)^{2q}\big]\,,
\end{align*}
where $ Z^{h,i}_t$ is defined by $ Z^{h,i}_t= \sup_{s \leq t} \big| Y^{h,i}_t \big|$.
Therefore it is sufficient to bound component by component.
For all $i \in \left[|1,m\right|]$, by the triangle inequality, 
\begin{multline}\label{sup}
\EE\big[\big|Z^{h,i}_t\big| ^{2q}\big] \leq C_{16} \big( \EE \big[ \l X_0 \r ^{2q} \big]+\hfill \\ \hfill+ \EE\big[\big( \sup_{s \leq t} \big| \int_0^s \sum_{k=0}^{\lfloor \T/h\rfloor}  \mathds{1}_{\left] kh,(k+1)h \right]} (u) \big(b^i(Y_{kh}^h) + \eta^{(h),i} (Y_{kh}^h,W_{(k+1)h} - W_{kh})\big)\, du \big| \big)^{2q}\,\big] +\\
\hfill+\EE\big[\big( \sup_{s \leq t} \big| \int_0^s \sum_{k=0}^{\lfloor \T/h\rfloor}  \sum_{j=0}^{d}\mathds{1}_{\left] kh,(k+1)h \right]}(u)\sigma^{i,j}(Y_{kh}^h)\, d W_u^j\big| \big)^{2q}\, \big]\big)\,.
\end{multline}
Consider the second term in this inequality. By H\"older's inequality, 
\begin{multline*}\big| \int_0^s \sum_{k=0}^{\lfloor \T/h\rfloor}  \mathds{1}_{\left] kh,(k+1)h \right]} (u) \big(b^i(Y_{kh}^h) + \eta^{(h),i} (Y_{kh}^h,W_{(k+1)h} - W_{kh})\big)\, du \big|^{2q} \leq \hfill \\
\hfill \leq s^{2q-1} \int_0^s \big| \sum_{k=0}^{\lfloor \T/h\rfloor}  \mathds{1}_{\left] kh,(k+1)h \right]} (u) \big(b^i(Y_{kh}^h) + \eta^{(h),i} (Y_{kh}^h,W_{(k+1)h} - W_{kh})\big)\big|^{2q}\, du\,.
\end{multline*}
Note that this sum over $k$ is reduced to one term for each $s$. Hence,
\begin{multline*}
\int_0^s \big| \sum_{k=0}^{\lfloor \T/h\rfloor}  \mathds{1}_{\left] kh,(k+1)h \right]} (u) \big(b^i(Y_{kh}^h) + \eta^{(h),i} (Y_{kh}^h,W_{(k+1)h} - W_{kh})\big)\big|^{2q}\, du\,. \leq \hfill \\
\hfill \leq C_{17}\int_0^s \sum_{k=0}^{\lfloor \T/h\rfloor}  \mathds{1}_{\left] kh,(k+1)h \right]} (u)\big( \big| b^i(Y_{kh}^h) \big|^{2q}+\big| \eta^{(h),i} (Y_{kh}^h,W_{(k+1)h} - W_{kh})\big|^{2q}\big)\, du\,.
\end{multline*}
From Assumption (H3) and the linear growth bound on the function $b$,
\begin{multline*}
 \int_0^s \sum_{k=0}^{\lfloor \T/h\rfloor}  \mathds{1}_{\left] kh,(k+1)h \right]} (u)\big( \big| b^i(Y_{kh}^h) \big|^{2q}+\big| \eta^{(h),i} (Y_{kh}^h,W_{(k+1)h} - W_{kh})\big|^{2q}\big)\, du\leq \hfill \\
 \hfill \leq C_{18} \int_0^s \sum_{k=0}^{\lfloor \T/h\rfloor}  \mathds{1}_{\left] kh,(k+1)h \right]} (u)\big( 1 +  \big(Z_u^h\big)^{2q} + h^{\alpha 2q}\big( Z_u^h \big)^{2q} + \big| W_{(k+1)h} - W_{kh}\big|^{2q}\big)\, du\,.
\end{multline*}
Since we are interested in $h$ small, we can consider $h \leq 1$. Therefore, since $t \leq \T$, we obtain
\begin{multline*}
\EE\big[\big( \sup_{s \leq t} \big| \int_0^s \sum_{k=0}^{\lfloor \T/h\rfloor}  \mathds{1}_{\left] kh,(k+1)h \right]} (u) \big(b^i(Y_{kh}^h) + \eta^{(h),i} (Y_{kh}^h,W_{(k+1)h} - W_{kh})\big)\, du \big| \big)^{2q}\,\big] \leq \hfill \\
\hfill \leq C_{19} (1 + \int_0^t  \EE\big[\big( Z_u^h \big)^{2q}\,\big]  +\EE\big[\sum_{k=0}^{\lfloor \T/h\rfloor}  \mathds{1}_{\left] kh,(k+1)h \right]}(u) \big| W_{(k+1)h} - W_{kh}\big|^{2q}\,\big] \, du\,.
\end{multline*}
Recall that for each $u$ the sum over $k$ is reduced to one term. Since each term can be bounded by $h^q(-\log h)^q$ as previously shown, then, for $h$ sufficiently small, 
\begin{multline*}
\EE\big[\big( \sup_{s \leq t} \big| \int_0^s \sum_{k=0}^{\lfloor \T/h\rfloor}  \mathds{1}_{\left] kh,(k+1)h \right]} (u) \big(b^i(Y_{kh}^h) + \eta^{(h),i} (Y_{kh}^h,W_{(k+1)h} - W_{kh})\big)\, du \big| \big)^{2q}\,\big] \leq \hfill \\
\hfill \leq C_{20} \big( 1 + \int_0^t  \EE\big[\big( Z_u^h \big)^{2q}\,\big] \, du \big) \,.
\end{multline*}
Consider now the last term in Inequality (\ref{sup}). By Burkh\"older Inequality, we get 
\begin{multline*}
\EE\big[\big( \sup_{s \leq t} \big| \int_0^s \sum_{k=0}^{\lfloor \T/h\rfloor}  \sum_{j=0}^{d}\mathds{1}_{\left] kh,(k+1)h \right]}(u)\sigma^{i,j}(Y_{kh}^h)\, d W_u^j\big| \big)^{2q}\, \big] \leq \hfill \\
\hfill \leq C_{21} t^{q-1} \big( \int_0^t \EE\big[\sum_{k=0}^{\lfloor \T/h\rfloor}  \sum_{j=0}^{d}\mathds{1}_{\left] kh,(k+1)h \right]}(u)\big|\sigma^{i,j}(Y_{kh}^h)\big|^{2q}\,\big]\,du \big)\,.
\end{multline*}
From the the linear growth bound of $\s$, we obtain
\begin{multline*}
\EE\big[\big( \sup_{s \leq t} \big| \int_0^s \sum_{k=0}^{\lfloor \T/h\rfloor}  \sum_{j=0}^{d}\mathds{1}_{\left] kh,(k+1)h \right]}(u)\sigma^{i,j}(Y_{kh}^h)\, d W_u^j\big| \big)^{2q}\, \big] \leq \hfill \\
 \hfill \leq C_{22} \big( 1+ \int_0^t \EE\big[\big( Z^h_u\big)^{2q}\,\big]\,du \big)\,.
\end{multline*}
Finally, we get
$$
\EE\big[\big|Z^{h,i}_t\big| ^{2q}\big] \leq C_{23 }\big( 1+ \int_0^t \EE\big[\big( Z^h_u\big)^{2q}\,\big]\,du \big)\,,
$$
and then,
$$
\EE\big[\big(Z^{h}_t\big) ^{2q}\big] \leq C_{24}\big( 1+ \int_0^t \EE\big[\big( Z^h_u\big)^{2q}\,\big]\,du \big)\,.
$$
Hence, by Gronwall's lemma,
$$\EE\big[\big(Z^{h}_{\T}\big) ^{2q}\big] \leq C_{25}\,,$$
where $C_{25}$ is independent of $h$.
\end{proof}

\begin{proof}{Theorem 5.2.}
For all positive $\T$, the same strategy as in the proof of Lemma \ref{le4} is set up to show the convergence on the time interval $\big[0, \T]$.

The error between the solution $X_t$ of the SDE and the process $X_t^h$ is denoted by~$\epsilon_t$.
Let us begin with a formula which lies the errors at two consecutive points of $h\NN$,
\begin{align*}
\epsilon_{(n+1)h} =& X_{(n+1)h} - X_{(n+1)h}^h\\
=&\epsilon_{nh} +\int_{nh}^{(n+1)h} \sigma(X_s) - \sigma(X_{nh}^h)\, d W_s + \\
 &\quad \ +\int_{nh}^{(n+1)h} b(X_s)- b(X_{nh}^h)-\eta^h (X_{nh}^h, W_{(n+1)h}-W_{nh})\, ds\,.
\end{align*}
As previously seen, Lemma \ref{le2} is applied to the process $(X_{nh+t} - Y^h_{nh+t})$ instead of $\epsilon_{nh+t}$ because of the linear interpolation in the definition of $X_t^h$. 
Then, we get
\begin{align*} \mathbb{E}\big[\big| \epsilon_{(n+1)h} \big|^{2q}\big] \leq &\mathbb{E}\big[\left| \epsilon_{nh} \right|^{2q}\big] +  C_0
 \int_{nh}^{(n+1)h} \big{(} \mathbb{E}\big[ \left| X_{s} - Y^h_{s} \right|^{2q}\big]+ \\ 
 &\quad+ \mathbb{E}\big[ \|\sigma(X_s)-\sigma(X_{nh}^h)\|^{2q}\big]  +\\
 &\quad+ \mathbb{E}\big[\left|b(X_s)- b(X_{nh}^h) - \eta^h (X_{nh}^h, W_{(n+1)h}-W_{nh}) \right|^{2q}\big]\big{)}\, ds\,.
\end{align*}
Or,
\[X_{s} - Y^h_{s}  = \epsilon_{nh} + X_{s} - X_{nh} + Y_{nh}^h - Y_{s}^h\,.\]
Then, we obtain
\[\mathbb{E}\big[\left| X_{s} - Y^h_{s}  \right|^{2q}\big] \leq C_1 \big(\mathbb{E}\big[\left| \epsilon_{nh} \right|^{2q}\big] + \mathbb{E}\big[\left| X_{s} - X_{nh}\right|^{2q}\big] + \mathbb{E}\big[\left| Y_{s}^h - Y_{nh}^h\right|^{2q}\big]\big)\,.\]
From Lemma \ref{le3},
\[\mathbb{E}\big[\left| X_{s} - X_{nh}\right|^{2q}\big] \leq C_2 (1+ \mathbb{E}\big[\left| X_0 \right|^{2q}\big]) (s-nh)^q
\leq C_3 h^q\]
For the process $Y^h_t$, Lemma \ref{le4} gives the inequality
\[\mathbb{E}\big[\left| Y^h_{s} - Y^h_{nh}\right|^{2q}\big] \leq C_4 (h^{2q} + h^q(-\log h)^q)\,.\]
Therefore, for a small $h$,
\begin{align*}
\mathbb{E}\big[ \left| X_{s} - Y^h_{s} \right|^{2q}\big] \leq C_5 (\mathbb{E}\big[\left| \epsilon_{nh} \right|^{2q}\big] + h^q(-\log h)^q)\,.
\end{align*}
Since $\sigma$ is a Lipschitz function,
\begin{align*}
\mathbb{E}\big[\|\sigma(X_s)-\sigma(X_{nh}^h)\|^{2q}\big] &\leq K_0^{2q}\ \mathbb{E}\big[ \left|X_s-X_{nh}^h\right|^{2q}\big]\\
&\leq K_0^{2q}\ \mathbb{E}\big[\left|X_s- X_{nh} + X_{nh} - X_{nh}^h\right|^{2q}\big]\\
&\leq C_6 (\mathbb{E}\big[\left| \epsilon_{nh} \right|^{2q}\big] + h^q)\,.
\end{align*}
On the other hand, the last term can be also bounded
\begin{multline*}
\mathbb{E}\big[\left|b(X_s)- b(X_{nh}^h) - \eta^h (X_{nh}^h, W_{(n+1)h}-W_{nh}) \right|^{2q}\big] \hfill \\ 
\hfill \leq C_7 \big(\mathbb{E}\big[\left|b(X_s)- b(X_{nh}^h)\right|^{2q}\big] + 
 \mathbb{E}\big[\left|\eta^h (X_{nh}^h, W_{(n+1)h}-W_{nh}) \right|^{2q}\big]\big)\,.
\end{multline*}
As $\mathbb{E}\big[\left| X_{nh}^h \right|^{2q}\big]$ is finite, 
\[\mathbb{E}\big[\left|\eta^h (X_{nh}^h, W_{(n+1)h}-W_{nh}) \right|^{2q}\big] \leq K (h^q(-\log h)^q + h^{2q\alpha})\,.\]
And, finally, we have
\begin{multline*}
\mathbb{E}\big[\left|b(X_s)- b(X_{nh}^h) - \eta^h (X_{nh}^h, W_{(n+1)h}-W_{nh}) \right|^{2q}\big] \leq C_8 (\mathbb{E}\big[\left| \epsilon_{nh} \right|^{2q}\big]+\hfill\\
\hfill + h^q(-\log h)^q + h^{2q\alpha})\,.
\end{multline*}
Thus, for $h$ sufficiently small,
\begin{multline*} 
\mathbb{E}\big[\left| \epsilon_{(n+1)h} \right|^{2q}\big]\leq \mathbb{E}\big[\left| \epsilon_{nh} \right|^{2q}\big] + C_0 h\Big{(}C_5 (\mathbb{E}\big[\left| \epsilon_{nh} \right|^{2q}\big] + h^q(-\log h)^q)+ \hfill \\
\hfill + C_6 (\mathbb{E}\big[\left| \epsilon_{nh} \right|^{2q}\big] + h^q) +C_8 (\mathbb{E}\big[\left| \epsilon_{nh} \right|^{2q}\big] + h^q(-\log h)^q + h^{2q\alpha})\Big{)}\\
\phantom{\mathbb{E}\big[\left| \epsilon_{(n+1)h} \right|^{2q}\big]\ \ \ } \leq (1 + C_9h)\mathbb{E}\big[\left| \epsilon_{nh} \right|^{2q}\big] + C_{10} h^{q+1}(-\log h)^q + C_{11} h^{1+2q\alpha}  \,.\hfill
\end{multline*}
Note that $(\mathbb{E}\big[\left| \epsilon_{nh} \right|^{2q}\big])$ is an other subarithmetico-geometric sequence.
Hence, as $\epsilon_0 =0$,
\begin{align*}
\mathbb{E}\left[\left| \epsilon_{nh} \right|^{2q}\right] \leq& n e^{nC_9 h} (C_{10} h^{q+1}(-\log h)^q + C_{11} h^{1+2q\alpha})\,.
\end{align*}
On the other hand,
\[\epsilon_t= \epsilon_{\lfloor t/h \rfloor h} + X_s-X_{nh} + X_{nh}^h- X_s^h\,.\]
Therefore, for a time $t$ in $[0,\T]$,
\begin{align}
\mathbb{E}\big[\left| \epsilon_t \right|^{2q}\big] \leq& C_5 (\mathbb{E}\left[\left| \epsilon_{\lfloor t/h \rfloor h} \right|^{2q}\right] + h^q(-\log h)^q )\nonumber \\
\leq& C_5 \big[ \frac{t}{h} e^{t/h C_9 h} (C_{10} h^{q+1}(-\log h)^q + C_{11} h^{1+2q\alpha}) + h^q (-\log h)^q\big] \nonumber \\
\leq& C_{12}( h^q(-\log h)^q + h^{2q\alpha})\,. \label{51}
\end{align}
Eventually, a bound on $\mathbb{E}\big[\left| X_t - X_t^h \right|^{2q}\big]$ is found.
However, we want to prove the inequality on the supremum
\[\mathbb{E}\big{[}\big( \sup_{t \in [0,\T]} \left| \epsilon_t \right|\big)^{2q}\big{]} \leq C (h^{2q\alpha} + h^q (-\log h)^q)\,.\]
Let us proceed with the proof of this inequality.
From the definition of the norm and the convexity of the application $x \longmapsto \left| x \right|^p$, 
\begin{align*}
\mathbb{E}\big[\big(\sup_{t \in [0,\T]} \left| \epsilon_t \right|\big)^{2q}\big] \leq C_{13} \sum_{i=1}^m \mathbb{E}\big[\big(\sup_{t \in [0,\T]} \left|\epsilon_t^i\right|\big)^{2q}\big]\,.
\end{align*}
Thus, it is sufficient to control the supremum of each component.

For all $i$ in $[|1,m|]$,
\begin{align*}
\epsilon_t^i =& \int_0^t \sum_{k=0}^{\lfloor \T/h\rfloor} \big[b^i(X_s)-b^i(X_{kh}^h) - (\eta^{(h)})^i (X_{kh}^h,W_{(k+1)h} - W_{kh})\big] \mathds{1}_{\left] kh,(k+1)h \right]} (s)\, ds\\
 &+ \int_0^t \sum_{k=0}^{\lfloor \T/h\rfloor} \sum_{j=1}^d \big[\sigma^{i,j}(X_s)-\sigma^{i,j}(X_{kh}^h)\big] \mathds{1}_{\left] kh,(k+1)h \right]} (s)\, d W^j_s\\
&+ \sum_{j=1}^d \sigma^{i,j}(X_{\lfloor t/h\rfloor h}^h)\big[W^j_t - W^j_{\lfloor t/h\rfloor h} - \frac{t - \lfloor t/h\rfloor h}{h} (W^j_{(\lfloor t/h\rfloor +1)h} - W^j_{\lfloor t/h\rfloor h})\big]\,.
\end{align*}
Then,
\begin{multline}
\mathbb{E}\big{[}\big(\sup_{t \in [0,\T]} \left|\epsilon_t^i\right|\big)^{2q}\big{]} \hfill \\
\leq C_{14}\big( \mathbb{E}\big{[}\big(\sup_{t \in [0,\T]} \big| \int_0^t \sum_{k=0}^{\lfloor \T/h\rfloor} \big\{b^i(X_s)-b^i(X_{kh}^h)+\hfill \\
 \hfill - (\eta^{(h)})^i (X_{kh}^h,W_{(k+1)h} - W_{kh})\big\} \mathds{1}_{\left] kh,(k+1)h \right]} (s)\, ds \big|\big)^{2q} \big{]} \\
+\mathbb{E} \big{[}\big(\sup_{t \in [0,\T]} \big| \int_0^t \sum_{k=0}^{\lfloor \T/h\rfloor} \sum_{j=1}^d \big\{\sigma^{i,j}(X_s)-\sigma^{i,j}(X_{kh}^h)\big\} \mathds{1}_{\left] kh,(k+1)h \right]} (s)\, d W^j_s \big|\big)^{2q} \big{]} \hfill \\
+ \mathbb{E}\big{[}\big(\sup_{t \in [0,\T]} \big|\sum_{j=1}^d \sigma^{i,j}(X_{\lfloor t/h\rfloor h}^h)\big\{W^j_t - W^j_{\lfloor t/h\rfloor h} + \hfill \\
\hfill- \frac{t - \lfloor t/h\rfloor h}{h} (W^j_{(\lfloor t/h\rfloor +1)h} - W^j_{\lfloor t/h\rfloor h})\big\}\big|\big)^{2q} \big{]} \big)\,.\label{52}
\end{multline}
Consider the first term of (\ref{52}), by H\"older's inequality, we get
\begin{multline*}
\mathbb{E}\big{[}\big(\sup_{t \in [0,\T ]} \big| \int_0^t \sum_{k=0}^{\lfloor \T/h\rfloor} \big\{b^i(X_s)-b^i(X_{kh}^h) + \hfill \\
\hfill- (\eta^{(h)})^i (X_{kh}^h,W_{(k+1)h} - W_{kh})\big\} \mathds{1}_{\left] kh,(k+1)h \right]} (s)\, ds \big\}\big|\big)^{2q}\big{]} \\
\leq C_{15}\ \mathbb{E}\big{[}\int_0^{\T} \big| \sum_{k=0}^{\lfloor \T /h\rfloor} \big\{b^i(X_s)-b^i(X_{kh}^h) +\hfill \\
\hfill - (\eta^{(h)})^i (X_{kh}^h,W_{(k+1)h} - W_{kh})\big\} \mathds{1}_{\left] kh,(k+1)h \right]} (s)\big|^{2q}\, ds\big{]}\,.
\end{multline*}
Note that, for all $s$, the sum over $k$ is just composed by only one term, thus
\begin{multline*}
\mathbb{E}\big{[}\sup_{t \in [0,\T]} \big| \int_0^t \sum_{k=0}^{\lfloor \T /h\rfloor} \big\{b^i(X_s)-b^i(X_{kh}^h)+\hfill \\
\hfill - (\eta^{(h)})^i (X_{kh}^h,W_{(k+1)h} - W_{kh})\big\} \mathds{1}_{\left] kh,(k+1)h \right]} (s)\, ds \big|^{2q}\big{]}\\
\leq C_{15} \int_0^{\T } \mathbb{E}\big{[}\sum_{k=0}^{\lfloor \T /h\rfloor} \big| b^i(X_s)-b^i(X_{kh}^h) +\hfill \\
\hfill- (\eta^{(h)})^i (X_{kh}^h,W_{(k+1)h} - W_{kh})\big|^{2q} \mathds{1}_{\left] kh,(k+1)h \right]} (s)\big{]}\, ds\,.
\end{multline*}
Or, from previous inequalities on processes,
\begin{multline*}
\mathbb{E}\big{[}\sup_{t \in [0,\T]}  \big| \int_0^t \sum_{k=0}^{\lfloor \T /h\rfloor} \big\{b^i(X_s)-b^i(X_{kh}^h)+\hfill \\
\hfill - (\eta^{(h)})^i (X_{kh}^h,W_{(k+1)h} - W_{kh})\big\} \mathds{1}_{\left] kh,(k+1)h \right]} (s)\, ds \big|^{2q}\big{]} \\
\hfill \leq C_{16} ( h^q(-\log h)^q + h^{2q\alpha})\,.
\end{multline*}
Consider now the second term of (\ref{52}). By Burkholder's inequality, we obtain 
\begin{align*}
\mathbb{E}\big{[} &\big(\sup_{t \in [0,\T ]} \big| \int_0^t  \sum_{k=0}^{\lfloor \T/h\rfloor} \sum_{j=1}^d (\sigma^{i,j}(X_s)-\sigma^{i,j}(X_{kh}^h)) \mathds{1}_{\left] kh,(k+1)h \right]} (s)\, d W^j_s \big|\big)^{2q}\big{]}\\
&\quad\leq C_{17}\mathbb{E}\big{[} \int_0^{\T} \big| \sum_{k=0}^{\lfloor \T /h\rfloor} \sum_{j=1}^d (\sigma^{i,j}(X_s)-\sigma^{i,j}(X_{kh}^h)) \mathds{1}_{\left] kh,(k+1)h \right]} (s)\big|^{2q} \,ds \big{]}\,\\
&\quad \leq C_{17} \int_0^{\T} \EE\big[ \sum_{k=0}^{\lfloor \T /h\rfloor} \big| \sum_{j=1}^d (\sigma^{i,j}(X_s)-\sigma^{i,j}(X_{kh}^h))\big|^{2q}\mathds{1}_{\left] kh,(k+1)h \right]} (s)\big]\,ds \,.
\end{align*}
Therefore, from Assumption (H1) and the bound on $\epsilon_t$,
\begin{multline*}
\mathbb{E}\big{[}\big( \sup_{t \in [0,\T ]} \big{|} \int_0^t  \sum_{k=0}^{\lfloor \T /h\rfloor} \sum_{j=1}^d (\sigma^{i,j}(X_s)-\sigma^{i,j}(X_{kh}^h)) \mathds{1}_{\left] kh,(k+1)h \right]} (s)\, d W^j_s \big{|}\big)^{2q} \big{]} \hfill \\
\hfill \leq C_{18}(h^q(-\log h)^q + h^{2q\alpha})\,.
\end{multline*}
Let us proceed with the remaining term of (\ref{52}). Note that
\begin{multline*}
\mathbb{E}\big{[}\big(\sup_{t \in [0,\T]} \big|\sum_{j=1}^d \sigma^{i,j}(X_{\lfloor t/h\rfloor h}^h)\big\{W^j_t - W^j_{\lfloor t/h\rfloor h} + \hfill \\
\hfill- \frac{t - \lfloor t/h\rfloor h}{h} (W^j_{(\lfloor t/h\rfloor +1)h} - W^j_{\lfloor t/h\rfloor h})\big\}\big|\big)^{2q} \big{]}\leq \\
\hfill \leq C_{19}\mathbb{E}\big{[}\sup_{t \in [0,\T]} \big(\sum_{j=1}^d \big|\sigma^{i,j}(X_{\lfloor t/h\rfloor h}^h)\big|^{2q} \sup_{s \in [\lfloor t/h\rfloor h,(\lfloor t/h\rfloor +1) h]}\big|W^j_s - W^j_{\lfloor t/h\rfloor h}\big|^{2q}\big)\big{]}\\
\hfill \leq C_{19}\mathbb{E}\big{[}\sup_{k \in [|0,\lfloor \T /h \rfloor|]} \big(\sum_{j=1}^d \big|\sigma^{i,j}(X_{kh}^h)\big|^{2q} \sup_{s \in [k h,(k +1) h]}\big|W^j_s - W^j_{kh}\big|^{2q}\big)\big{]}
\end{multline*}
Since the supremum of non-negative elements is less than the sum of this elements and from the linear growth bound on $\s$, we get
\begin{multline*}
\mathbb{E}\big{[}\sup_{k \in [|0,\lfloor \T /h \rfloor|]} \big(\sum_{j=1}^d \big|\sigma^{i,j}(X_{kh}^h)\big|^{2q} \sup_{s \in [k h,(k +1) h]}\big|W^j_s - W^j_{kh}\big|^{2q}\big)\big{]} \leq \hfill \\
\hfill \leq C_{20} \mathbb{E}\big{[} \sum_{k=1}^{\lfloor \T /h \rfloor} \big(1+ \big|X_{kh}^h\big|^{2q}\big) \sup_{s \in [k h,(k +1) h]}\big|W^j_s - W^j_{kh}\big|^{2q} \big{]}
\end{multline*}
Since $ \mathbb{E}\big[ \underset{t \in [0,\T ]}{\sup} \big|\sigma^{i,j}(X_{\lfloor t/h\rfloor h}^h)\big|^{2q}\big]$ is finite (Lemma \ref{le4}), by independence of increments of Brownian motion and, from the bound on the supremum of these supremum on a time interval of length $h$, we obtain
\begin{multline*}
\mathbb{E}\big{[}\big(\sup_{t \in [0,\T]} \big|\sum_{j=1}^d \sigma^{i,j}(X_{\lfloor t/h\rfloor h}^h)\big\{W^j_t - W^j_{\lfloor t/h\rfloor h} + \hfill \\
\hfill- \frac{t - \lfloor t/h\rfloor h}{h} (W^j_{(\lfloor t/h\rfloor +1)h} - W^j_{\lfloor t/h\rfloor h})\big\}\big|\big)^{2q} \big{]}\leq \\
\hfill \leq C_{21} \lfloor \T /h\rfloor h^{q} (-\log h)^q\,.
\end{multline*}
Finally, we obtain
\begin{align*}
\mathbb{E}\big[ \big( \sup_{t \in [0,\T]} \left|\epsilon_t^i\right|\big)^{2q}\big] \leq C_{22} ( h^{q-1} (-\log h)^q + h^{2q\alpha})\,.
\end{align*}
And, therefore, we have
\[\mathbb{E}\big[  \big(\sup_{t \in [0,\T]} \left| \epsilon_t \right|\big)^{2q}\big] \leq C_{22} ( h^{q-1} (-\log h)^q + h^{2q\alpha})\,.\]
This inequality implies the convergence in $L^{p}$ of the process $(X^{h}_t)$ to $(X_t)$ on $[0,\T]$ when $h$ tends to $0$.

\smallskip
For the almost sure convergence, we suppose also that the exponent $p$ is such that $p>2$ and $p \alpha >1$.

Thus, the previous result gives
\[\sum_{k=1}^{\infty} \mathbb{E}\big[\big(\sup_{t \in [0,\T]} \big|X_t - X_t^{1/k}\big|\big)^{2q}\big] < \infty\,.\]
As $\big(\sup_{t \in [0,\T]} \big|X_t - X_t^{1/k}\big|\big)^{2q}$ are non-negative random variables, then
\[\lim_{h\rightarrow 0} \sup_ {t \in [0,\T]} \left|\epsilon_t\right| =0 \hspace{3mm} \text{a.s.}\,.\]
\end{proof}
From the convergence of this scheme of stochastic numerical analysis, the main result on the convergence of the process $\br{X}^h_t$ to a solution of a SDE is deduced.
\begin{theo}\label{Main}
Suppose that there exist measurable maps $b$, $\s$, and $\eta^{(h)}$ which verify Assumption (H1) for $b$ and $\s$, and (H4) for $\eta^{(h)}$ such that the application on the first component of $\td{T}^{(h)}$ is of the following form,
\[ U^{(h)}(x,y) = x+ \sigma(x)y+ hb(x)+h \eta^{(h)}(x,y)\,.\]
For all $x_0$ in $\RR^m$, and all $\T>0$, let $X_t^{x_0}$ be the solution on $[0,\T]$ of the SDE
\[dX^{x_0}_t= b(X^{x_0}_t)dt + \sigma(X_t^{x_0})dW_t\,.\]
Then, the process $(\br{X}^h_t)$, starting in $x_0$, converges to $(X_t^{x_0})$ when $h$ tends to $0$ in $L^{p}$, for all $p>2$ on $\left[0, \T\right]$.

\smallskip
Moreover, the convergence is almost sure on $\left[0, \T\right]$.
\end{theo}
We now present a similar result in the case of locally Lipschitz and linearly bounded applications $b$ and $\s$, that is to say satisfying Assumptions (H2) and (H3).

\begin{theo}\label{locally}
For all $\T>0$, under the assumptions (H1), (H2), and (H3), the process $(X_t^h)$ converges in $L^{p}$ to the solution $(X_t)$ of the SDE \eqref{eq:A} on $\left[0,\T\right]$ for all $p > 2$.

More precisely,
\[\lim_{h \rightarrow 0}\mathbb{E}\big[ (\sup_{t \in [0,\T]} \big| X_t -X_t^h\big|)^p \big] = 0\,.\]
\end{theo}

\begin{proof}
For all $\T>0$ and all natural $N$, we define two stopping times $\T_N$ and $\T_N^{(h)}$ by 
\[\T_N = \inf\{t; \left| X_t \right| \geq N \} \ \text{ and }\  \T_N^{(h)} = \inf\{t; \big| X_t^{(h)} \big| \geq N \}\,.\]
Note that, from Theorem \ref{globally},
\[\EE\big[ \big(\sup_{t \in \left[0,\T\right]} \big|X^h_{t \wedge \T_N^{(h)} \wedge \T_N} - X_{t \wedge \T_N^{(h)} \wedge \T_N} \big|)^{2q} \big] \leq C_N (h^{q-1}(-\log h)^q+ h^{2q\alpha})\,,\]
where $C_N$ notably depends on $K_N$.
\begin{multline*}
\EE\big[\big(\sup_{t \in \left[0,\T\right]} \big|X^h_{t} - X_{t} \big|)^{2q} \big]\leq \EE\big[\big(\sup_{t \in \left[0,\T\right]} \big|X^h_{t \wedge \T_N^{(h)} \wedge \T_N} - X_{t \wedge \T_N^{(h)} \wedge \T_N} \big|)^{2q} \big]+\hfill \\
\hfill+\EE\big[\big(\sup_{t \in \left[0,\T\right]} \big|X^h_{t}\big|^{2q}+\big| X_{t} \big|^{2q} : \T_N^{(h)} \leq \T \big] +\\
\hfill+ \EE\big[\big(\sup_{t \in \left[0,\T\right]} \big|X^h_{t}\big|^{2q}+\big| X_{t} \big|^{2q} : \T_N \leq \T \big]\,.
\end{multline*}
Hence, we get for all $N$
\begin{multline*}
\EE\big[\big(\sup_{t \in \left[0,\T\right]} \big|X^h_{t} - X_{t} \big|)^{2q} \big]\leq \EE\big[\sup_{t \in \left[0,\T\right]} \big|X^h_{t \wedge \T_N^{(h)} \wedge \T_N} - X_{t \wedge \T_N^{(h)} \wedge \T_N} \big|^{2q} \big]+\hfill \\
\hfill+\dfrac2N \EE\big[\sup_{t \in \left[0,\T\right]} \big|X^h_{t}\big|^{2q+1}+\sup_{t \in \left[0,\T\right]} \big| X_{t} \big|^{2q+1} \big]\,.
\end{multline*}
Since the expectation of the supremum of $X_t$ (Lemma \ref{le2}) and the supremum of $X_t^h$ (Lemma \ref{le3}) can be bounded by a constant independent of $h$ for all $h$ sufficiently small, then
\[\lim_{h \rightarrow 0}\mathbb{E}\big[ (\sup_{t \in [0,\T]} \big| X_t -X_t^h\big|)^{2q} \big] = 0\,.\]
\end{proof}
This theorem \ref{locally} implies the following result on the limit evolution of the system when the time step $h$ for the interactions goes to $0$.
\begin{theo}\label{Main2}
Suppose that there exist measurable maps $b$, $\s$, and $\eta^{(h)}$ which verify Assumptions (H2) and (H3) for $b$ and $\s$, and (H4) for $\eta^{(h)}$ such that the application on the first component of $\td{T}^{(h)}$ is of the following form,
\[ U^{(h)}(x,y) = x+ \sigma(x)y+ hb(x)+h \eta^{(h)}(x,y)\,.\]
For all $x_0$ in $\RR^m$, and all $\T>0$, let $X_t^{x_0}$ be the solution on $[0,\T]$ of the SDE
\[dX^{x_0}_t= b(X^{x_0}_t)dt + \sigma(X_t^{x_0})dW_t\,.\]
Then, the process $(\br{X}^h_t)$, starting in $x_0$, converges to $(X_t^{x_0})$ when $h$ tends to $0$ in $L^{p}$, for all $p>2$ on $\left[0, \T\right]$.
\end{theo}
Physically, the results of these two theorems of convergence can be understood as follows. If the effective action of the environment on the system is roughly linear, the limit evolution of the system is given by a solution of a stochastic differential equation. This SDE is deduced from a Taylor expansion of the application $U^{(h)}$.

With the quantum repeated interactions scheme, Attal and Pautrat (in \cite{AP}) find quantum Langevin equations as limits of some Hamiltonian systems. There are some similarities with their results, particularly on the form of considered interactions.

\section{Back to the Physical Systems}
\subsection{Charged Particle in an Uniform Electric Field}
The first example was a charged particule in an uniform electric field.
Recall that the evolution for a time step $h$ is given by
\[X(nh)= U^{(h)}(X((n-1)h),E(nh))\,,\]
where the map $U^{(h)}$ is defined by
\[ U^{(h)}(x,y)= x +  \s(x) y +h b(x) + h\eta^{h}(x,y)\,,\]
with
\[\s \left(\begin{array}{c} x_1\\ x_2 \end{array}\right)=\left(\begin{array}{c} 0\\ q \end{array}\right) \hspace{10mm}b\left(\begin{array}{c} x_1\\ x_2 \end{array}\right)=\left(\begin{array}{c} \dfrac{x_2}{m}\\ 0 \end{array}\right)\hspace{10mm}\eta^{h}(x,y)=\left(\begin{array}{c} \dfrac{qy}{2m}\\ 0 \end{array}\right)\,.\]
Note that this application $U^{(h)}$ is of the form of Theorem \ref{Main}.
The function $\s$ is constant and the function $b$ is linear. Finally, they are Lipschitz functions.
For the application $\eta^{h}$, Assumption (H4) is also verified with $\alpha = + \infty$ and $K_2 = \dfrac{q}{2m}$.
Therefore, Theorem \ref{Main} can be applied for this system.

For all intial $x_0$ and $p_0$, all $\T>0$, the limit process which gives the evolutions of the charged particule is almost surely the solution $X_t=\left(\begin{array}{c}
X_t^1\\
X_t^2
\end{array}\right)$ on $[0,\T]$ of the stochastic differential equation
\[dX_t =\left(\begin{array}{c}
\frac{X_t^2}{m}\\
0
\end{array}\right) dt + \left(\begin{array}{c}
0\\
q
\end{array} \right)dW_t,\]
where $W_t$ is a $1$-dimensional standard Brownian motion and with $X_0 = (x_0,p_0)$.

\subsection{Harmonic Interaction}

The second example was a harmonic interaction between the system and the environment.
The evolution of the system was descibed by the Markov chain
\[X(nh)=U^{(h)}(X((n-1)h), Y(nh))\,\]
 where
\begin{align*}
U^{(h)}(X,Y)= X + \sigma(X) Y + h b(X) +h \eta^{(h)}(X,Y)\,,
\end{align*}
with
\[b\left(\begin{array}{c} x_1\\x_2  \end{array} \right) = \left(\begin{array}{c} x_2\\ -(x_1+l) \end{array} \right), \hspace{5mm} \sigma \left(\begin{array}{c} x_1\\ x_2 \end{array} \right) =\left(\begin{array}{cc} 0& 0\\ 1& 0 \end{array} \right)\,,\]
and
\[\eta^{(h)}\left[\left(\begin{array}{c} x_1\\ x_2 \end{array}\right), \left(\begin{array}{c} y_1\\ y_2 \end{array}\right) \right] = \dfrac{1}{2}\left(\begin{array}{c} y_1\\ y_2 \end{array}\right)-\dfrac{h}{2} \left(\begin{array}{c} x_1+l-y_2/3  \\ x_2+2y_1/3 \end{array}\right) + \circ(h)\,.\]
The functions $b$ and $\s$ are Lipschitz.
For the application $\eta^{(h)}$, Assumptions (H4) is verified too for $\alpha=1$.
Theorem \ref{Main} can be applied for this system.
Therefore, we conclude that for all $\T>0$ and all $Q_1(0)$, $P_1(0)$ in $\mathbb{R}$, the limit evolution on $\left[0,\T \right]$ of the system is given almost surely by the solution of the stochastic differential equation 
 \[dX_t = \left( \begin{array}{c}
X_t^2\\
-(X^1_t +l)
\end{array} \right)dt +\left(\begin{array}{cc}
0 &0\\
1& 0
\end{array} \right)dW_t\,,\]
starting at time $0$ in $X_0 = \left(\begin{array}{c} Q_1(0)\\ P_1(0) \end{array}\right)$.

Note that this SDE is the equation of a harmonic oscillator perturbated by a Brownian noise. This kind of SDE was already considered in \cite{T} for instance.

\smallskip
Let us focus now on the asymptotic behaviour of this process. This stochastic differential equation has no stationaty measure since the non-perturbated differential equation has no stable point. Physically, the energy of the system whose evolution is governed by this SDE inscreases with the time. This rise of energy can be explained by the fact that repeated interactions bring energy to the system and there is no way to dissipate it. The next example shall be different.

\subsection{Damped Harmonic Oscillator}

The last example is the damped harmonic oscillator. The system undergoing repeated interactions evolves following the Markov chain $(X(nh))$ defined by
\[X(nh)=U^{(h)}(X((n-1)h), Y(nh))\,\]
 where
\begin{align*}
U^{(h)}(X,Y)= X + \sigma(X) Y + h b(X) +h \eta^{(h)}(X,Y),
\end{align*}
with
\[b\left(\begin{array}{c} x_1\\x_2  \end{array} \right) = \left(\begin{array}{c} x_2\\ -x_1 - f x_2 \end{array} \right), \hspace{5mm} \sigma \left(\begin{array}{c} x_1\\ x_2 \end{array} \right) =\left(\begin{array}{cc} 0& 0\\ 1& 0 \end{array} \right)\,.\]
The application $\eta^{(h)}$ is not explicitly expressed. But as said in the section 4, from Newton's law of motion, high derivative of $P_1$ and $Q_1$ can be bounded with $Q_1$, $P_1$, $Q_2$ et $P_2$. Then a required bound on $\eta^{(h)}$ can be found. Moreover, the maps $b$ et $\s$ are Lipschitz. Therefore Theorem \ref{Main} can be applied.

Hence, the limit evolution of the system is governed by the solution of the stochastic differential equation
\[\begin{cases}
 d\, Q_1 = P_1 dt\\
 d\,P_1 = - f \,P_1 - Q_1 dt + d\,W_t^1\,.
  \end{cases}\]
Note that in all these examples, the states of parts of the environment are sampled from the increments of a Brownian motion whose variance is $1$.
A new parameter which could be interpreted as the temperature of the environment can be added by taking a variance $T$. Mathematically, it can be introduced by multipling states of the environment in the setup of repeated interactions by a factor $\sqrt{T}$.
Physically, more the temperature of the environment is high, more the interaction with the environment influences the dynamics of the system.
The effect of this temperature is the change of the diffucion term. The new stochastic differential equation is  
\[\begin{cases}
 d\, Q_1 = P_1 dt\\
 d\,P_1 = - f \,P_1 - Q_1 dt + \sqrt{T}\, d\,W_t^1\,.
  \end{cases}\]
Then, we examine the asymptotic behaviour of the system.
This stochastic differential equation is a Langevin equation whose stationary measure is given by the Gibbs measure
$$ d\,\mu = \frac{e^{-f \frac{(Q_1)^2+(P_1)^2}{2T}}}{\mathcal{Z}}\,d\,Q_1\,d\,P_1\,,$$
where $\mathcal{Z}$ is a normalizing constant.

Contrary to the previous example, the friction force allows the system to dissipate a part of its energy and then, the convergence of the dynamics to the stationary state (see \cite{MSH}). Physically, this convergence to this Gibbs measure can be understood as follows. The repeated interactions of the environment at temperature $T$ lead the system to tend to the thermal state related to the same temperature $T$. Finally, the system is thermalised by the environment.

\end{document}